\documentclass[jounal]{IEEEtran}
\usepackage{cite}
\usepackage{graphicx}
\usepackage{amsmath}
\usepackage{amssymb}
\usepackage{amsthm}
\usepackage{amsfonts}
\usepackage{algorithm}
\usepackage[noend]{algpseudocode}
\usepackage{array}
\usepackage{color}
\usepackage{enumerate}
\usepackage{subfigure}

\makeatletter

\usepackage{algorithm,algpseudocode,float}
\usepackage{lipsum}



\begin{document}

\title{HPE Transformer:
Learning to Optimize Multi-Group Multicast Beamforming Under Nonconvex QoS Constraints}

\author{Yang~Li
        ~and~Ya-Feng~Liu

\thanks{Part of this work was presented at the IEEE Globecom 2023 \cite{Li_Liu23}.}
\thanks{Y. Li is with the Shenzhen Research Institute of Big Data,
Shenzhen 518172, China
(e-mail: liyang@sribd.cn).}
\thanks{Y.-F. Liu is with
the State Key Laboratory of Scientific and Engineering
Computing, Institute of Computational Mathematics and
Scientific/Engineering Computing, Academy of Mathematics and Systems Science,
Chinese Academy of Sciences, Beijing 100190, China (e-mail: yafliu@lsec.cc.ac.cn).}}

\maketitle
\begin{abstract}
This paper studies the quality-of-service (QoS) constrained
multi-group multicast beamforming \textcolor{black}{design} problem,
where each multicast group is composed of a number of users requiring the same content.
Due to the nonconvex QoS constraints, this problem is nonconvex and NP-hard. While existing optimization-based iterative algorithms
can obtain a suboptimal solution, their iterative nature results in large computational complexity and delay.
To facilitate real-time implementations, this paper proposes a
deep learning-based approach, which consists of
a beamforming structure assisted problem transformation and a customized neural network architecture named
hierarchical permutation equivariance (HPE) transformer.
The proposed HPE transformer is proved to be
permutation equivariant with respect to the users within each multicast group,
and also permutation equivariant with respect to different multicast groups.
Simulation results demonstrate that the proposed HPE transformer outperforms state-of-the-art
optimization-based and
deep learning-based approaches for multi-group multicast beamforming design
in terms of the total transmit power, the constraint violation, and the computational time.
In addition, the proposed HPE transformer achieves pretty good generalization performance on different numbers of
users, different numbers of multicast groups, and different signal-to-interference-plus-noise ratio targets.
\end{abstract}



\section{Introduction}
Multicast beamforming has been recognized as an effective transmission technique
for content delivery, in which the users requiring the same content
are partitioned into a multicast group so that the
base station (BS) can serve them using the same beam simultaneously \cite{Sidiropoulos2006}.
This technique has greatly improved both the spectrum and the energy efficiency
of multiuser multi-antenna wireless networks.
Moreover, recent studies have demonstrated its great potential in
improving content distribution and delivery in content-centric cloud radio
access networks~\cite{Tao2016,Li2018,Dai2018,Li20}.

\subsection{Prior Works}
Multicast beamforming design was initially considered for only a single multicast
group \cite{Abdelkader,Kim2011,Wu2013,Lu2017}, and then was extended to a more general multi-group scenario \cite{Karipidis2008,Chang08,Christopoulos,Xiang2013}.
In addition, multicast beamforming design has also been investigated for non-orthogonal
multicast and unicast transmissions \cite{ChenEr2018,Li2019-2,ZZ23}.
One of the most important optimization problems in
multicast beamforming design is to minimize the total transmit power
at the BS subject to a minimum signal-to-interference-plus-noise ratio (SINR) target for each user,
which is referred to as the quality-of-service (QoS) constrained multicast beamforming design problem.
It has been shown in \cite{Sidiropoulos2006} that this problem
is nonconvex and generally NP-hard.
To obtain a suboptimal solution of the QoS constrained multicast beamforming design problem
within polynomial time complexity, a great deal of effort has been made in the literature.
The existing approaches can be roughly divided into
optimization-based and deep learning-based approaches.

\subsubsection{Optimization-Based Approaches}
There are two main categories of optimization-based approaches for solving the QoS constrained multicast beamforming design problem~\cite{liu2024survey}.
One category of approaches is the semi-definite relaxation (SDR) \cite{Sidiropoulos2006},
which relaxes the original problem into a semi-definite programming formulation
\cite{Karipidis2008,Chang08,Christopoulos,Xiang2013},
and then solves the relaxed problem using the interior-point method.
Theoretical analysis shows that the SDR-based approaches can achieve provable approximation
accuracy \cite{Luo07}. However, the SDR requires to square the number of variables, making
its computational complexity grow dramatically
as the problem size increases,
and the resultant performance also deteriorates severely as the number of users increases \cite{Karipidis2008}.

To reduce the computational complexity while maintaining the solution quality, another
category of approaches called the convex-concave procedure (CCP) \cite{Yuille2003}
(also known as the successive convex approximation)
has been proposed for multicast beamforming design \cite{Tran14,Christopoulos15,Sadeghi17,Chen2017,Mohamadi22,ZhangDong22}.
Instead of leveraging a single relaxed problem,
the CCP tackles the original problem by
successively solving a sequence of convex subproblems,
where each convex subproblem is constructed by linearizing the concave part of the nonconvex QoS constraints. Although the CCP is guaranteed to return a stationary point of the
original problem, its computational
complexity is still high, as many convex subproblems are typically solved by the
interior-point method \cite{Tran14,Christopoulos15}.
To further reduce the computational complexity,
the CCP can be combined with the suboptimal zero-forcing scheme \cite{Sadeghi17} or
the alternating direction method of multipliers \cite{Chen2017,Mohamadi22,ZhangDong22}.
While the CCP achieves better performance and lower computational complexity than the SDR,
both of them belong to optimization-based iterative algorithms,
which make them difficult to support real-time implementations.

\subsubsection{Deep Learning-Based Approaches}
Recently, deep learning has shown its great potential in obtaining
real-time solutions for radio resource management \cite{SunHR2018,Liang20,Eisen2020,Shen2021,Guo2021,Yunqi22,Zhang22,Xiao21}.
\textcolor{black}{These approaches learn the input-output mapping function by training a neural network,
which is applicable to various wireless channel
states/realizations
subject to the same distribution.
When the wireless channel states vary, the well-trained neural network can still be used as long as the channel distribution does not change.
Therefore, when the channel distribution does not change,
the real-time execution of the deep learning-based approaches does not include the computationally expensive re-training,
but only needs to re-execute the simple feed-forward operations,
whose computational complexity and delay is much lower than that of the optimization-based iterative algorithms.}
A typical example is to train a fully connected neural network for
learning the mapping function from wireless channels to the solution of beamforming design
and power allocation \cite{SunHR2018,Liang20}.
However, since fully connected neural networks cannot exploit the graph topology in the wireless networks,
they usually require quite a large number of training samples to achieve satisfactory performance.
Moreover, since their input/output dimension depends on a pre-defined problem size,
they need to be retrained whenever the number of users varies.
To incorporate the graph topology in the wireless networks and make the neural network architecture
scalable to different input/ouput dimensions,
recent works investigated the application of the graph neural networks (GNNs) for beamforming design and power allocation
in wireless networks \cite{Eisen2020,Shen2021,Guo2021,Yunqi22}.
\textcolor{black}{As it is generally hard to get the labels for radio resource management,
a prevalent solution is to apply unsupervised learning \cite{Liang20,Eisen2020,Shen2021,Guo2021,Yunqi22,Zhang22} or reinforcement learning \cite{Xiao21},
where the loss function or reward
is constructed directly from the optimization objective.}

While there are many existing \textcolor{black}{deep learning-based approaches for radio resource management}, 
to the best of our knowledge, there is only one recent work proposing a deep learning-based approach
\textcolor{black}{for multi-group multicast beamforming design~\cite{Zhang22}},
which extends the existing GNNs from the unicast beamforming design
to the more complicated max-min fairness multi-group multicast beamforming design.
\textcolor{black}{However, a} direct application of \textcolor{black}{the multi-group multicast GNN (MMGNN)
in \cite{Zhang22}} for the QoS constrained \textcolor{black}{multi-group multicast beamforming design}
will result in poor performance due to the following two reasons.
First, the QoS constraints are nonconvex and more complicated than the peak power constraint in the max-min fairness problem.
Due to the lack of a constraint augmented mechanism in the neural network architecture,
MMGNN cannot well satisfy the nonconvex QoS constraints.
Second, it is highly expected that
the neural network architecture can well generalize on both different numbers of users and different numbers of multicast groups.
Nevertheless, 
MMGNN shares the trainable parameters across different users, making it only well generalize on different numbers of users but have to be retrained once the number of multicast groups varies.

Another powerful neural network architecture is the transformer \cite{Vaswani2017}, which originates from the natural
language processing (NLP).
Actually, the transformer inherits all the desired properties of GNNs and hence can be viewed as a special case of GNNs for fully connected graphs \cite{Ye,Li_HT}.
Specifically, based on an attention mechanism, the transformer can well model the interactions across different nodes of
the input graph, making it quite qualified to learn the relevance among different input components \cite{Li_HT,Li21}.
Furthermore, the self-attention block in the transformer can effectively guarantee the permutation equivariance (PE) property
as in GNNs. In particular, if the indices of any input components
are exchanged, the self-attention block can automatically produce a corresponding permutation
of the original output. This PE property can reduce the trainable
parameter space and avoid a large number of unnecessary permuted training samples \cite{Eisen2020,Shen2021,Guo2021,Yunqi22}.
Last but not the least, due to the sharing of the trainable parameters across different graph nodes,
the well-trained transformer can well generalize on different problem sizes.

\subsection{Main Contributions}
\textcolor{black}{This paper investigates how to design a deep learning-based approach
that can establish the mapping function from wireless channels to the corresponding multicast beamformers.
Therefore, different from the classical optimization-based approaches, the
proposed deep learning-based approach has an extremely short computational time
due to the fast feed-forward operations, making it a more practical solution for real-time implementations.}
Yet instead of directly constructing a neural network for representing the mapping function
from wireless channels to multicast beamformers,
we exploit an inherent multicast beamforming structure \cite{Dong20},
which can effectively decompose the desired mapping function into two separate mapping functions.
The first mapping function is from wireless channels to low-dimensional parameters of each user,
and the second mapping function is from the low-dimensional parameters to multicast beamformers.
Different from the basic PE property, we observe that the two mapping functions inherently have
a \emph{hierarchical} permutation equivariance (HPE)
property, i.e., permutation equivariant with respect to the users within each multicast group, and also permutation equivariant with respect to different multicast groups.


The HPE property provides an effective guidance for the neural network architecture design,
\textcolor{black}{which has also been recently considered for time series forecasting \cite{Umagami23}, where two parallel self-attention blocks are used to extract the relevance among the series in the same class
and among different classes, respectively.
However, it requires to
combine the representation of different series in the same
class before performing self-attention among different classes, and then
repeats the self-attention output so that each series in the same class
has the same representation. Different from the parallel
structure in \cite{Umagami23}, this paper proposes an HPE transformer with
a two-level hierarchical structure,
which does not need to combine and repeat the representation of different users,
making it preserve the individual information of each user
and maintain the desired PE property with respect to different users in the same multicast group.
Specifically,} 
we first propose an encoding block for representing the first mapping function,
consisting of an embedding layer, multiple hierarchical layers, and a de-embedding layer.
Since the encoding block leverages a hierarchical transformer-based structure, we prove its HPE property
based on the basic PE property of the transformer.
Furthermore, we propose a decoding block for representing the second mapping function,
including a solution construction layer and multiple constraint augmented layers.
In particular, the solution construction layer is built on the semi-closed form solution \cite{Dong20} given by the multicast beamforming structure,
and the constraint augmented layers are constructed by unfolding a number of gradient descent steps for deceasing the constraint violation.
The HPE property of the decoding block is also proved based on the multicast beamforming structure.

The main contributions of this work are summarized as follows.
\begin{enumerate}[1)]
\item
\textcolor{black}{We propose a novel perspective on solving the QoS constrained multi-group multicast beamforming design problem, which
incorporates the closed-form solution structure into the neural network architecture design.
Based on the solution structure,
this paper decomposes the mapping function from wireless channels to multicast beamformers
into two separate mapping functions, both of which enjoy the HPE property.
The HPE property provides a key guidance on how to design the neural network architecture.}

\item
We further propose an HPE transformer with an encoding block and a decoding block to represent the corresponding two mapping functions, respectively.
Different from the vanilla transformer, the proposed HPE transformer is proved to enjoy the more structured HPE property, which is desired in the multi-group multicast beamforming design.
This is in sharp contrast to the basic PE property, which is even enforced with respect to any input components no matter whether
they are from the same group. Moreover, due to the judicious design of the decoding block, the nonconvex QoS constraints can be well respected.
This is quite different from most of the existing works for beamforming design subject to simple power constraints.

\item
Simulation results show that the proposed HPE transformer achieves much lower transmit power and
much smaller constraint violation
than state-of-the-art deep learning-based approaches for multi-group multicast beamforming design. Moreover, due to the fast feed-forward operations,
the computational time of the proposed HPE transformer is much less than that of the optimization-based iterative algorithms. Finally,
the proposed HPE transformer is demonstrated with pretty good generalization performance under different
numbers of users, different numbers of multicast groups, and different SINR targets.
\end{enumerate}

The remainder of this paper is organized as follows. System model and problem formulation
are introduced in Section II.
An HPE property is revealed based on the inherent multicast beamforming structure in Section III.
To guarantee the HPE property, an HPE transformer is designed in Section IV.
Simulation results are provided in Section V. Finally,
Section VI concludes the paper.

Throughout this paper, scalars, vectors, and matrices are denoted by lower-case letters, lower-case bold letters, and upper-case bold letters, respectively. The real and complex domains are denoted by $\mathbb{R}$ and $\mathbb{C}$, respectively.
We denote the transpose, conjugate transpose, inverse,
real part, and imaginary part
of a vector/matrix by $(\cdot)^\mathrm{T}$, $(\cdot)^\mathrm{H}$, $(\cdot)^{-1}$, $\Re(\cdot)$, and $\Im(\cdot)$, respectively.
The $N\times N$ identity matrix and the length-$N$ all-one vector are
denoted as $\mathbf{I}_N$ and $\mathbf{1}_N$, respectively.
The column vectorization of a matrix is denoted as $\text{vec}(\cdot)$, the operation
$\otimes$ denotes the Kronecker product, and $\text{ReLU}(\cdot)$ denotes $\max(\cdot,0)$.

\section{System Model and Problem Formulation}
Consider a downlink multiuser system, where
a BS equipped with $N$ antennas serves $M$
multicast groups of users. Each group $m\in\mathcal{M}\triangleq\{1,2,...,M\}$
consists of $K_m$ single-antenna users, which require a common message from the BS.
The set of the user indices in group $m$
is denoted as $\mathcal{K}_m\triangleq\left\{1,2,...,K_m\right\}$ for all $m\in\mathcal{M}$.
Different multicast groups are disjoint, which means that
the total number of users is $K=\sum_{m=1}^M K_m$.

\textcolor{black}{We adopt a quasi-static block fading
channel model, where the channels between the BS
and users remain constant within each coherence block, but
may vary among different coherence blocks.}
Let $\mathbf{h}_{m,k}\in\mathbb{C}^N$ denote the
channel vector from the BS to
the $k$-th user in the $m$-th multicast group.
Let $s_m$ denote the multicast message intended for
the $m$-th multicast group with $\mathbb{E}\left[
\vert s_m \vert^2\right]=1$,
and $\mathbf{w}_m\in\mathbb{C}^N$ denote the
corresponding multicast beamforming vector.
\textcolor{black}{The BS multicasts the common message $s_m$
to the users in the $m$-th multicast group for $m\in\mathcal{M}$,
and the $M$ different messages are transmitted with the $M$
corresponding beamformers simultaneously.}
The received signal at the $k$-th user in the $m$-th multicast group is given by
\begin{eqnarray}\label{received signal}
y_{m,k} =
\mathbf{h}_{m,k}^{\mathrm{H}}\mathbf{w}_ms_m+
\sum_{j=1,j\neq m}^{M} \mathbf{h}_{m,k}^{\mathrm{H}}\mathbf{w}_js_j
+n_{m,k},\nonumber\\ \forall~ k\in \mathcal{K}_m,~~\forall~ m\in\mathcal{M},
\end{eqnarray}
where $n_{m,k}$ is the additive white Gaussian
noise with zero mean and variance $\sigma_{m,k}^2$.
Correspondingly, the total transmit power at
the BS is given by $\sum_{m=1}^M\left\Vert\mathbf{w}_{m}\right\Vert_2^2$, and
the SINR of
the $k$-th user in the $m$-th multicast group
can be written as
\begin{eqnarray}\label{SINR}
\text{SINR}_{m,k} =
\frac{\left\vert\mathbf{h}_{m,k}^{\mathrm{H}}\mathbf{w}_m\right\vert^2}
{\sum_{j=1,j\neq m}^{M} \left\vert\mathbf{h}_{m,k}^{\mathrm{H}}\mathbf{w}_j\right\vert^2+\sigma_{m,k}^2},
\nonumber\\ \forall~ k\in \mathcal{K}_m,~~\forall~ m\in\mathcal{M}.
\end{eqnarray}
To minimize the total transmit power while satisfying
the QoS constraint of each user, the multi-group multicast beamforming design
problem can be formulated as\textcolor{black}{{\footnote{\textcolor{black}{It is interesting to consider the robust beamforming design with the imperfect channel state information when the perfect channel state information is hard to obtain \cite{ZZ23}. We leave the investigation of
the problem in the case of
the imperfect channel information as the future work.}}}}
\begin{subequations}\label{MGMB}
\begin{equation}\label{MGMB_power}
\min_{\mathbf{W}}~
\sum_{m=1}^M\left\Vert\mathbf{w}_{m}\right\Vert_2^2~~~~~~~~~~~~~~~~~~~~~~~~~~~~~~~~~
\end{equation}
\begin{eqnarray}\label{MGMB_SINR}
\text{s.t.}\
~~\frac{\left\vert\mathbf{h}_{m,k}^{\mathrm{H}}\mathbf{w}_m\right\vert^2}
{\sum_{j=1,j\neq m}^{M} \left\vert\mathbf{h}_{m,k}^{\mathrm{H}}\mathbf{w}_j\right\vert^2+\sigma_{m,k}^2}\geq
\gamma_m,\nonumber\\
\forall~ k\in \mathcal{K}_m,~~\forall~ m\in\mathcal{M},
\end{eqnarray}
\end{subequations}
where $\mathbf{W}\triangleq\left[\mathbf{w}_1, \mathbf{w}_2, \ldots, \mathbf{w}_M\right]$
and $\gamma_m$ is the SINR target of the users in the $m$-th multicast group.

Problem \eqref{MGMB} is nonconvex due to the
non-convexity of the QoS constraints in \eqref{MGMB_SINR}.
In fact, it is even NP-hard to find a feasible solution satisfying
\eqref{MGMB_SINR} \cite{Murty1987}.
Existing approaches for solving problem \eqref{MGMB}
include the SDR \cite{Sidiropoulos2006}
and the CCP \cite{Yuille2003}, which however lead to squaring the number of variables
or successively solving a sequence of non-trivial convex optimization problems.
Both of the above approaches require to solve problem \eqref{MGMB} with a lot of iterations
in an instance-by-instance manner.
Once the wireless channels vary, these optimization-based iterative algorithms have to be
re-executed, which induces tremendous computational complexity and delay.

In the following two sections,
instead of pursuing optimization-based iterative algorithms,
we propose a deep learning-based method
for solving problem \eqref{MGMB} by constructing a mapping function
from wireless channels to multicast beamformers.
Once the mapping function is obtained, it can infer the multicast
beamformers corresponding to any wireless channels by fast feed-forward computations.
The proposed method consists of a
beamforming structure assisted problem transformation (Section~III)
and a customized neural network architecture (Section~IV).

\section{Beamforming Structure with HPE \textcolor{black}{Property}}
In this section, we exploit an inherent multicast beamforming structure to
establish a problem transformation for the original problem \eqref{MGMB}.
\textcolor{black}{The transformed problem reveals an HPE
property, which provides a key guidance on how to design the neural network architecture.
This observation motivates us to design the HPE transformer that well incorporates the HPE property, making the proposed deep learning-based approach significantly outperform the existing works
(as shown in Section V).}

\subsection{Problem Transformation Based on Beamforming Structure}
While problem \eqref{MGMB} is nonconvex and NP-hard,
its optimal solution can be written in the form of a weighted minimum mean squared error
filter \cite{Dong20}, which is given by
\begin{eqnarray}\label{structure}
\mathbf{w}_{m} = \left(\mathbf{I}_N+\sum_{j\in\mathcal{M}}
\sum_{k\in\mathcal{K}_j}
\lambda_{j,k}\gamma_{j}\mathbf{h}_{j,k}\mathbf{h}_{j,k}^\mathrm{H}\right)^{-1}\mathbf{H}_m\boldsymbol{\alpha}_m,\nonumber\\
\forall~ m\in\mathcal{M},
\end{eqnarray}
where $\mathbf{H}_m\triangleq
\left[\mathbf{h}_{m,1},\mathbf{h}_{m,2},\ldots,\mathbf{h}_{m,K_m}\right]$ is the channel matrix of the users in the $m$-th
multicast group,
$\boldsymbol{\alpha}_m\triangleq\left[\alpha_{m,1},\alpha_{m,2},\ldots,\alpha_{m,K_m}\right]^\mathrm{T}$ is a corresponding complex weight vector,
and $\boldsymbol{\lambda}_m\triangleq\left[\lambda_{m,1},\lambda_{m,2},\ldots,\lambda_{m,K_m}\right]^\mathrm{T}$
is a corresponding real vector.
Interestingly, $\mathbf{H}_m\boldsymbol{\alpha}_m$ in \eqref{structure}
can be viewed as a weighted combination of the user channels in the $m$-th
multicast group, where each element in $\boldsymbol{\alpha}_m$ corresponds to the weight of each user channel. Moreover, each $\lambda_{m,k}\geq 0$ can be interpreted as
the Lagrange multiplier corresponding to the SINR constraint of
the $k$-th user in the $m$-th multicast group.

The expression of $\mathbf{w}_{m}$ in \eqref{structure}
provides a semi-closed form solution for problem \eqref{MGMB}.
Although the unknown parameters
$\left\{\boldsymbol{\alpha}_m\right\}_{m\in\mathcal{M}}$
and $\left\{\lambda_{m,k}\right\}_{k\in\mathcal{K}_m,m\in\mathcal{M}}$
still need to be determined,
the parameter dimension is much smaller than
the original variable dimension especially when the number of antennas $N$ is large \cite{Dong20}.
Compared with the original problem \eqref{MGMB} with
$MN$ complex variables in $\mathbf{W}$ (i.e., $2MN$ real variables),
the semi-closed form solution in \eqref{structure} contains
only $K=\sum_{m=1}^M K_m$ complex parameters in $\left\{\boldsymbol{\alpha}_m\right\}_{m\in\mathcal{M}}$
and $K$ real parameters in $\left\{\lambda_{m,k}\right\}_{k\in\mathcal{K}_m,m\in\mathcal{M}}$,
reducing the total number of unknowns from $2MN$ to $3K$.

The inherent low-dimensional solution structure in \eqref{structure}
not only results in lower computational complexity
for obtaining the multicast beamforming solution when $N$ is large \cite{Dong20},
but also provides an effective guidance for constructing the mapping function
from wireless channels to multicast beamformers.
Let $\mathbf{H}\triangleq\left[\mathbf{H}_1,\mathbf{H}_2,\ldots,\mathbf{H}_M\right]$,
$\boldsymbol{\alpha}\triangleq\left[\boldsymbol{\alpha}_1^\mathrm{T},\boldsymbol{\alpha}_2^\mathrm{T},\ldots,\boldsymbol{\alpha}_M^\mathrm{T}\right]^\mathrm{T}$,
and $\boldsymbol{\lambda}\triangleq\left[\boldsymbol{\lambda}_1^\mathrm{T},\boldsymbol{\lambda}_2^\mathrm{T},\ldots,\boldsymbol{\lambda}_M^\mathrm{T}\right]^\mathrm{T}$,
respectively.
Regarding the parameters $\left(\boldsymbol{\alpha}, \boldsymbol{\lambda}\right)$ as intermediate variables,
the original problem \eqref{MGMB} can be reformulated as
\begin{subequations}\label{transform}
\begin{equation}\label{transform_power}
\min_{f(\cdot), g(\cdot, \cdot, \cdot)}~~
\sum_{m=1}^M\left\Vert\mathbf{w}_{m}\right\Vert_2^2~~~~~~~~~~~~~~~~~~~~~~~~~~~
\end{equation}
\begin{equation}\label{transform_f}
\text{s.t.}\
~~~~\left(\boldsymbol{\alpha}, \boldsymbol{\lambda}\right)=f\left(\mathbf{H}\right),~~~~~~~~~~~~~~~~~~~~
\end{equation}
\begin{equation}\label{transform_g}
~~~~~~~~~~~~~~~~~~~~~~~~~~\mathbf{W}=g\left(\mathbf{H}, \boldsymbol{\alpha}, \boldsymbol{\lambda}\right)~~\text{with}~~\eqref{MGMB_SINR}~~\text{satisfied},
\end{equation}
\end{subequations}
where $f(\cdot):\mathbb{C}^{N\times K}\rightarrow
\mathbb{C}^{K}\times\mathbb{R}^{K}_{+}$ is a mapping function
from $\mathbf{H}$ to $\left(\boldsymbol{\alpha}, \boldsymbol{\lambda}\right)$,
$g(\cdot, \cdot, \cdot): \mathbb{C}^{N\times K}\times
\mathbb{C}^{K}\times\mathbb{R}^{K}_{+}\rightarrow
\mathbb{C}^{N\times M}$ is a mapping function from
$\left(\mathbf{H}, \boldsymbol{\alpha}, \boldsymbol{\lambda}\right)$ to $\mathbf{W}$,
and the final output $\mathbf{W}$
satisfies all the SINR constraints in $\eqref{MGMB_SINR}$.

Problem \eqref{transform} decomposes the desired mapping function from $\mathbf{H}$ to $\mathbf{W}$
into two separate mappings, i.e., $f(\cdot)$ and $g(\cdot, \cdot, \cdot)$.
On one hand, the intermediate variables $\left(\boldsymbol{\alpha}, \boldsymbol{\lambda}\right)$
have a lower dimension than $\mathbf{W}$, making the
optimization of $f(\cdot)$ much easier than that of the
mapping function from $\mathbf{H}$ to $\mathbf{W}$.
On the other hand, the beamforming structure in \eqref{structure} also provides
an effective guidance for the optimization of $g(\cdot, \cdot, \cdot)$.
In particular, if the SINR constraints in $\eqref{MGMB_SINR}$ are not considered,
$g(\cdot, \cdot, \cdot)$ can be even directly given by \eqref{structure} without any further design.

\vspace{2mm}
\textcolor{black}{\emph{Remark 1:}
Unlike the WMMSE reformulation, which can be solved by unfolding
the iterative procedure of the WMMSE algorithm \cite{Chowdhury21},
the problem transformation \eqref{transform}
is nontrivial to solve by any existing algorithm, let alone unfolding one.
Therefore, instead of unfolding any existing algorithm, we propose an HPE transformer that incorporates an inherent HPE property of the transformed problem \eqref{transform} into the neural network architecture.}

\subsection{HPE Property}
In Section~IV, we will design the neural network architecture for representing
$f(\cdot)$ and $g(\cdot, \cdot, \cdot)$ in \eqref{transform}.
To facilitate the architecture design, we first examine an
HPE property in $f(\cdot)$ and $g(\cdot, \cdot, \cdot)$, respectively.

First, the mapping function from $\mathbf{H}$ to $\left(\boldsymbol{\alpha}, \boldsymbol{\lambda}\right)$
should be independent of both the user's and the group's indices.
To be specific, if the indices of any users within the same multicast group
are exchanged, and/or if the indices of any multicast groups are exchanged,
$f(\cdot)$ should produce a corresponding permutation of the original output.
In other words, the desired $f(\cdot)$ inherently has
an HPE property: 1) permutation equivariant with respect to the indices of the users within each multicast group;
2) and also permutation equivariant with respect to the indices of multicast groups.
This HPE property of $f(\cdot)$ is formally defined as follows.
\newtheorem{property}{Property}
\begin{property}\label{HPE_f}
(HPE property of $f(\cdot)$)
Let $\tilde{\mathbf{H}}_m\triangleq\left[\mathbf{h}_{m,\pi_m(1)},
\mathbf{h}_{m,\pi_m(2)},\ldots,\mathbf{h}_{m,\pi_m(K_m)}\right]$,
$\tilde{\boldsymbol{\alpha}}_m\triangleq\left[\alpha_{m,\pi_m(1)},
\alpha_{m,\pi_m(2)},\ldots,\alpha_{m,\pi_m(K_m)}\right]^\mathrm{T}$,
and $\tilde{\boldsymbol{\lambda}}_m\triangleq\left[\lambda_{m,\pi_m(1)},
\lambda_{m,\pi_m(2)},\ldots,\lambda_{m,\pi_m(K_m)}\right]^\mathrm{T}$,
where $\pi_m(\cdot)$ is a permutation of the indices in $\mathcal{K}_m$.
Furthermore, let $\tilde{\mathbf{H}}\triangleq\left[\tilde{\mathbf{H}}_{\pi_0(1)},
\tilde{\mathbf{H}}_{\pi_0(2)},\ldots,\tilde{\mathbf{H}}_{\pi_0(M)}\right]$,
$\tilde{\boldsymbol{\alpha}}\triangleq\left[\tilde{\boldsymbol{\alpha}}_{\pi_0(1)}^\mathrm{T},
\tilde{\boldsymbol{\alpha}}_{\pi_0(2)}^\mathrm{T},\ldots,\tilde{\boldsymbol{\alpha}}_{\pi_0(M)}^\mathrm{T}\right]^\mathrm{T}$,
and $\tilde{\boldsymbol{\lambda}}\triangleq\left[\tilde{\boldsymbol{\lambda}}_{\pi_0(1)}^\mathrm{T},
\tilde{\boldsymbol{\lambda}}_{\pi_0(2)}^\mathrm{T},\ldots,\tilde{\boldsymbol{\lambda}}_{\pi_0(M)}^\mathrm{T}\right]^\mathrm{T}$,
where $\pi_0(\cdot)$ is a permutation of the indices in $\mathcal{M}$.
The function $f(\cdot)$ is said to have the HPE property if
\begin{eqnarray}\label{HPE1}
\left(\tilde{\boldsymbol{\alpha}}, \tilde{\boldsymbol{\lambda}}\right)=f\left(\tilde{\mathbf{H}}\right),
~~\forall~ \pi_{m}(\cdot):\mathcal{K}_m\rightarrow \mathcal{K}_m,
\nonumber\\
\forall~ m\in\mathcal{M},~~
\forall~ \pi_{0}(\cdot):\mathcal{M}\rightarrow \mathcal{M},
\end{eqnarray}
as long as $\left(\boldsymbol{\alpha}, \boldsymbol{\lambda}\right)=f\left(\mathbf{H}\right)$.
\end{property}

Similarly, the mapping function $g(\cdot, \cdot, \cdot)$
from $\left(\mathbf{H}, \boldsymbol{\alpha}, \boldsymbol{\lambda}\right)$ to $\mathbf{W}$
also inherently has an HPE property\footnote{
Notice that $g(\cdot, \cdot, \cdot)$ is permutation invariant with respect to the users within each multicast group, which means that
its outputs will not change no matter how the order of the users within each multicast group changes. This is a little bit different
from the HPE property of $f(\cdot)$, where the order of the outputs will change in the same way as that of the users within each multicast group.}, which is formally defined as follows.
\begin{property}\label{HPE_g}
(HPE property of $g(\cdot, \cdot, \cdot)$)
Let $\tilde{\mathbf{W}}\triangleq\left[\mathbf{w}_{\pi_0(1)},\mathbf{w}_{\pi_0(2)},\ldots,\mathbf{w}_{\pi_0(M)}\right]$.
The function $g(\cdot, \cdot, \cdot)$ is said to have the
HPE property if
\begin{eqnarray}\label{HPE2}
\tilde{\mathbf{W}}=g\left(\tilde{\mathbf{H}}, \tilde{\boldsymbol{\alpha}}, \tilde{\boldsymbol{\lambda}}\right),
~~\forall~ \pi_{m}(\cdot):\mathcal{K}_m\rightarrow \mathcal{K}_m,\nonumber\\
\forall~ m\in\mathcal{M},~~
\forall~ \pi_{0}(\cdot):\mathcal{M}\rightarrow \mathcal{M},
\end{eqnarray}
as long as $\mathbf{W}=g\left(\mathbf{H}, \boldsymbol{\alpha}, \boldsymbol{\lambda}\right)$.
\end{property}
\vspace{2mm}


\section{HPE Transformer and Its Architecture Design}
In this section, we propose a customized neural network architecture for representing the mapping function
from $\mathbf{H}$ to $\mathbf{W}$ for problem \eqref{transform}.
As shown in Fig.~\ref{HPET}, the overall architecture consists of an encoding block
and a decoding block, which are used to
represent $f\left(\cdot\right)$ and $g\left(\cdot, \cdot, \cdot\right)$, respectively.
From Section III-B, we have known that both $f\left(\cdot\right)$ and $g\left(\cdot, \cdot, \cdot\right)$
inherently have the HPE property.
To incorporate this property into the neural network architecture,
we first briefly review the vanilla transformer, which provides a basic PE property.

\begin{figure}[t!]
\begin{center}
  \includegraphics[width=0.4\textwidth]{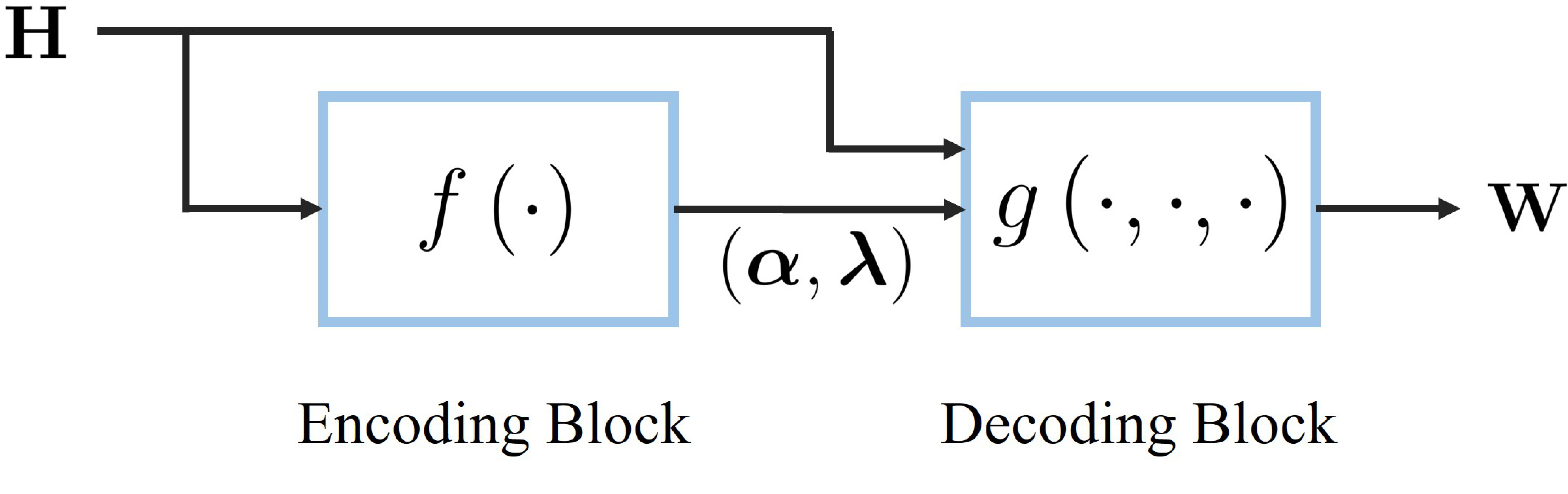}
  \caption{The overall architecture of the proposed HPE transformer $g(\cdot,f(\cdot))$, which consists of an encoding block
  for representing $f(\cdot)$ and a decoding block for representing $g(\cdot,\cdot,\cdot)$.}\label{HPET}
\end{center}
\vspace{-0.2cm}
\end{figure}

\subsection{Transformer Preliminaries}
The vanilla transformer is built on the self-attention block \cite{Vaswani2017},
which was originally designed for the NLP tasks such as machine translation.
The self-attention block not only extracts the relevance among different input components, but also guarantees the PE property.

As shown in Fig.~\ref{AB}, the self-attention block consists of
a multi-head attention (MHA) operation and a component-wise feed-forward (CFF) operation.
Let $\mathbf{X}\triangleq\left[\mathbf{x}_{1},\mathbf{x}_{2},\ldots,\mathbf{x}_{I}\right]\in
\mathbb{R}^{d\times I}$ denote $I$ input components,
\textcolor{black}{where $d$ denotes the embedding size of each component.}
The corresponding output $\mathbf{Z}\triangleq\left[\mathbf{z}_{1},\mathbf{z}_{2},\ldots,\mathbf{z}_{I}\right]\in\mathbb{R}^{d\times I}$
of the self-attention block is given by
\begin{equation}\label{A1}
\mathbf{Z}=\text{Norm}\left(
\mathbf{Y}+\text{CFF}\left(\mathbf{Y}\right)
\right),
\end{equation}
\begin{equation}\label{A2}
\text{with}~~\mathbf{Y}=\text{Norm}\left(
\mathbf{X}+\text{MHA}\left(\mathbf{X}\right)
\right),~~~~~
\end{equation}
where $\text{Norm}$ represents a normalization step \cite{Ba2016},
$\text{CFF}$ processes each column of $\mathbf{Y}\in\mathbb{R}^{d\times I}$ independently and identically
with a feed-forward operation,
and the operation of $\text{MHA}$ is shown in Appendix~\ref{MHA_expression}.
The above self-attention block given by \eqref{A1} and \eqref{A2}
satisfies the following property.
\begin{property}\label{PE}
(PE Property)
Let $u(\cdot):\mathbb{R}^{d\times I}\rightarrow\mathbb{R}^{d\times I}$ denote
the mapping function from $\mathbf{X}$ to $\mathbf{Z}$ given by \eqref{A1} and \eqref{A2}.
Define $\tilde{\mathbf{X}}\triangleq\left[\mathbf{x}_{\pi(1)},\mathbf{x}_{\pi(2)},\ldots,\mathbf{x}_{\pi(I)}\right]$
and $\tilde{\mathbf{Z}}\triangleq\left[\mathbf{z}_{\pi(2)},\mathbf{z}_{\pi(1)},\ldots,\mathbf{z}_{\pi(I)}\right]$,
where $\pi(\cdot)$ denotes a permutation of the indices in $\mathcal{I}\triangleq
\{1,2,\ldots,I\}$. Then, we always have
\begin{equation}\label{PEu}
\tilde{\mathbf{Z}} = u\left(\tilde{\mathbf{X}}\right),
~~\forall~ \pi(\cdot):\mathcal{I}\rightarrow \mathcal{I}.
\end{equation}
\end{property}
\begin{proof}
\textcolor{black}{Please refer to \cite[Proposition 1]{Li_HT}.}
\end{proof}
\vspace{2mm}

\begin{figure}[t!]
\begin{center}
  \includegraphics[width=0.48\textwidth]{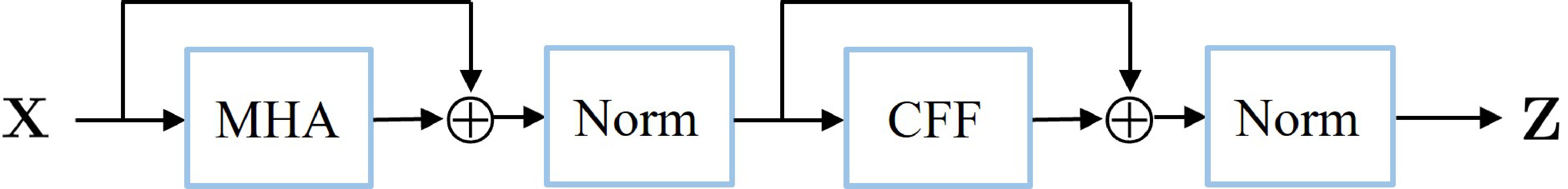}
  \caption{The self-attention block in the vanilla transformer.}\label{AB}
\end{center}
\vspace{-0.2cm}
\end{figure}

Property~\ref{PE} shows that the self-attention block in the vanilla transformer
satisfies a basic PE property.
In the next subsection, we will propose an encoding block for representing
the mapping function $f(\cdot)$, which requires a more structured HPE property.

\vspace{2mm}
\emph{Remark 2:}
The PE property given in \eqref{PEu} is actually a
special case of the more general HPE property in \eqref{HPE1}.
In particular, any mapping function satisfying \eqref{PEu} must also satisfy \eqref{HPE1}.
However, the PE property in \eqref{PEu} is too strict for problem \eqref{transform},
since we cannot enforce the PE property with respect to the users from different multicast groups.
For example, if the channel matrix is permuted as
$\left[\tilde{\mathbf{h}}_{1,1}, \tilde{\mathbf{h}}_{1,2}, \tilde{\mathbf{h}}_{2,1}, \tilde{\mathbf{h}}_{2,2}\right]=\left[\mathbf{h}_{1,1}, \mathbf{h}_{2,1}, \mathbf{h}_{1,2}, \mathbf{h}_{2,2}\right]$, the corresponding output
$\left[\tilde{\boldsymbol{\alpha}}_{1,1}, \tilde{\boldsymbol{\alpha}}_{1,2}, \tilde{\boldsymbol{\alpha}}_{2,1}, \tilde{\boldsymbol{\alpha}}_{2,2}\right]\neq\left[\boldsymbol{\alpha}_{1,1}, \boldsymbol{\alpha}_{2,1}, \boldsymbol{\alpha}_{1,2}, \boldsymbol{\alpha}_{2,2}\right]$.

\subsection{Encoding Block for Representing $f(\cdot)$}
The proposed encoding block for representing $f(\cdot)$
is composed of an embedding layer, $L$ hierarchical layers, and a de-embedding layer.
As shown in Fig.~~\ref{encoding}, the embedding layer transforms the channel matrix $\mathbf{H}$
into an initial embedding $\mathbf{X}^{(0)}\in\mathbb{R}^{d\times K}$,
where the embedding dimension $d$ is a hyper-parameter;
then, the initial embedding $\mathbf{X}^{(0)}$ is updated through
$L$ hierarchical layers to obtain
$\mathbf{X}^{(L)}\in\mathbb{R}^{d\times K}$; finally,
the de-embedding layer transforms $\mathbf{X}^{(L)}$ into
$\left(\boldsymbol{\alpha}, \boldsymbol{\lambda}\right)$.

\begin{figure}[t!]
\begin{center}
  \includegraphics[width=0.48\textwidth]{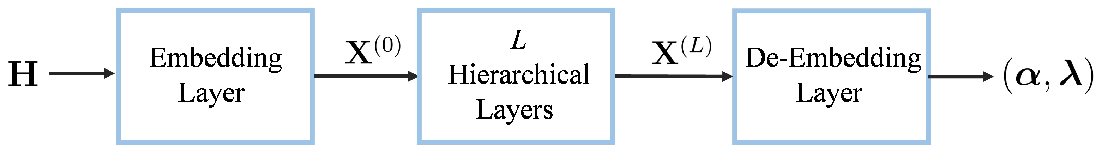}
  \caption{The architecture of the proposed encoding block for representing $f(\cdot)$.}\label{encoding}
\end{center}
\vspace{-0.2cm}
\end{figure}

\subsubsection{Embedding Layer}
The embedding layer first expresses the input $\mathbf{H}\in\mathbb{C}^{N\times K}$
as a real matrix by separating the real and imaginary parts:
\begin{equation}\label{EM1}
\mathbf{H}^{\text{em}}=
\left[{\Re\left(\mathbf{H}\right)}^\mathrm{T},{\Im\left(\mathbf{H}\right)}^\mathrm{T}\right]^\mathrm{T}.
\end{equation}
Then, the embedding layer obtains the initial embedding
$\mathbf{X}^{(0)}\in\mathbb{R}^{d\times K}$
by a linear projection:
\begin{eqnarray}\label{EM2}
\mathbf{X}^{(0)}=
\mathbf{W}^{\text{em}}\mathbf{H}^{\text{em}}+\mathbf{b}^{\text{em}}\otimes \mathbf{1}_{K}^\mathrm{T},
\end{eqnarray}
where $\mathbf{W}^{\text{em}}\in\mathbb{R}^{d\times 2N}$ and
$\mathbf{b}^{\text{em}}\in\mathbb{R}^{d}$
are the trainable matrix and vector, respectively.
In \eqref{EM2}, the initial embedding can be expanded as $\mathbf{X}^{(0)}=\left[\mathbf{X}^{(0)}_1,\mathbf{X}^{(0)}_2,\ldots,\mathbf{X}^{(0)}_M\right]$
and $\mathbf{X}^{(0)}_m=\left[\mathbf{x}^{(0)}_{m,1},\mathbf{x}^{(0)}_{m,2},\ldots,\mathbf{x}^{(0)}_{m,K_m}\right]$,
where $\mathbf{x}^{(0)}_{m,k}\in\mathbb{R}^d$ is the initial embedding of the channel vector $\mathbf{h}_{m,k}\in\mathbb{C}^N$.
It can be seen from $\eqref{EM2}$
that each $\mathbf{h}_{m,k}\in\mathbb{C}^N$ for all $k\in\mathcal{K}_m$ and $m\in\mathcal{M}$ is processed independently and identically
in the embedding layer.

\subsubsection{Hierarchical Layers}
The $\ell$-th hierarchical layer transforms $\mathbf{X}^{(\ell-1)}\in\mathbb{R}^{d\times K}$ into $\mathbf{X}^{(\ell)}\in\mathbb{R}^{d\times K}$, where
the input can be expanded as
$\mathbf{X}^{(\ell-1)}=\left[\mathbf{X}^{(\ell-1)}_1,\mathbf{X}^{(\ell-1)}_2,\ldots,\mathbf{X}^{(\ell-1)}_M\right]$
and $\mathbf{X}^{(\ell-1)}_m=\left[\mathbf{x}^{(\ell-1)}_{m,1},\mathbf{x}^{(\ell-1)}_{m,2},\ldots,\mathbf{x}^{(\ell-1)}_{m,K_m}\right]$.
As illustrated in Fig.~\ref{HIE}, in order
to guarantee the PE property with respect to the users within each multicast group,
we first adopt the self-attention block in the vanilla transformer to process each
$\mathbf{X}^{(\ell-1)}_m$ for all $m\in\mathcal{M}$.
With $u_1^{(\ell)}(\cdot)$
denoting the mapping function of the self-attention block,
we transform $\mathbf{X}^{(\ell-1)}_m\in\mathbb{R}^{d\times K_m}$ into $\mathbf{Z}^{(\ell)}_m\in\mathbb{R}^{d\times K_m}$ by
\begin{equation}\label{H1}
\mathbf{Z}^{(\ell)}_m=u_1^{(\ell)}\left(\mathbf{X}^{(\ell-1)}_m\right),~~\forall~ m\in\mathcal{M},~~\forall~ \ell\in\mathcal{L},
\end{equation}
where $u_1^{(\ell)}(\cdot)$ is shared among different multicast groups,
although the column dimension $K_m$ of the input $\mathbf{X}^{(\ell-1)}_m$ for different $m$ can be different.
This is because $u_1^{(\ell)}(\cdot)$ processes each column of the input matrix identically, making
its column dimension adaptable to different $K_m$.
Since each $\mathbf{X}^{(\ell-1)}_m$ for all $m\in\mathcal{M}$ is processed independently and identically in
\eqref{H1}, the proposed hierarchical layer is
permutation equivariant with respect to multicast groups and
also scale adaptable to different group numbers.

Furthermore, to allow the interaction among different multicast groups,
we adopt another self-attention block to process $\mathbf{Z}^{(\ell)}=\left[\mathbf{Z}^{(\ell)}_1,\mathbf{Z}^{(\ell)}_2,\ldots,\mathbf{Z}^{(\ell)}_M\right]$ by
\begin{equation}\label{H2}
\mathbf{X}^{(\ell)}=u_2^{(\ell)}\left(\mathbf{Z}^{(\ell)}\right),~~\forall~ \ell\in\mathcal{L},
\end{equation}
where
$\mathbf{X}^{(\ell)}=\left[\mathbf{X}^{(\ell)}_1,\mathbf{X}^{(\ell)}_2,\ldots,\mathbf{X}^{(\ell)}_M\right]$
and $\mathbf{X}^{(\ell)}_m=\left[\mathbf{x}^{(\ell)}_{m,1},\mathbf{x}^{(\ell)}_{m,2},\ldots,\mathbf{x}^{(\ell)}_{m,K_m}\right]$.
Compared with $u_1^{(\ell)}(\cdot)$, we use a different subscript for $u_2^{(\ell)}(\cdot)$,
because the trainable parameters in $u_2^{(\ell)}(\cdot)$ can be different from those in $u_1^{(\ell)}(\cdot)$. \textcolor{black}{Since $u_1^{(\ell)}(\cdot)$ guarantees the PE property with respect to the users within each multicast group, and $u_2^{(\ell)}(\cdot)$ further provides the PE property among different
multicast groups, then the hierarchical layers satisfy the HPE property,
whose formal proof is given in Appendix~C.}

\begin{figure}[t!]
\begin{center}
  \includegraphics[width=0.4\textwidth]{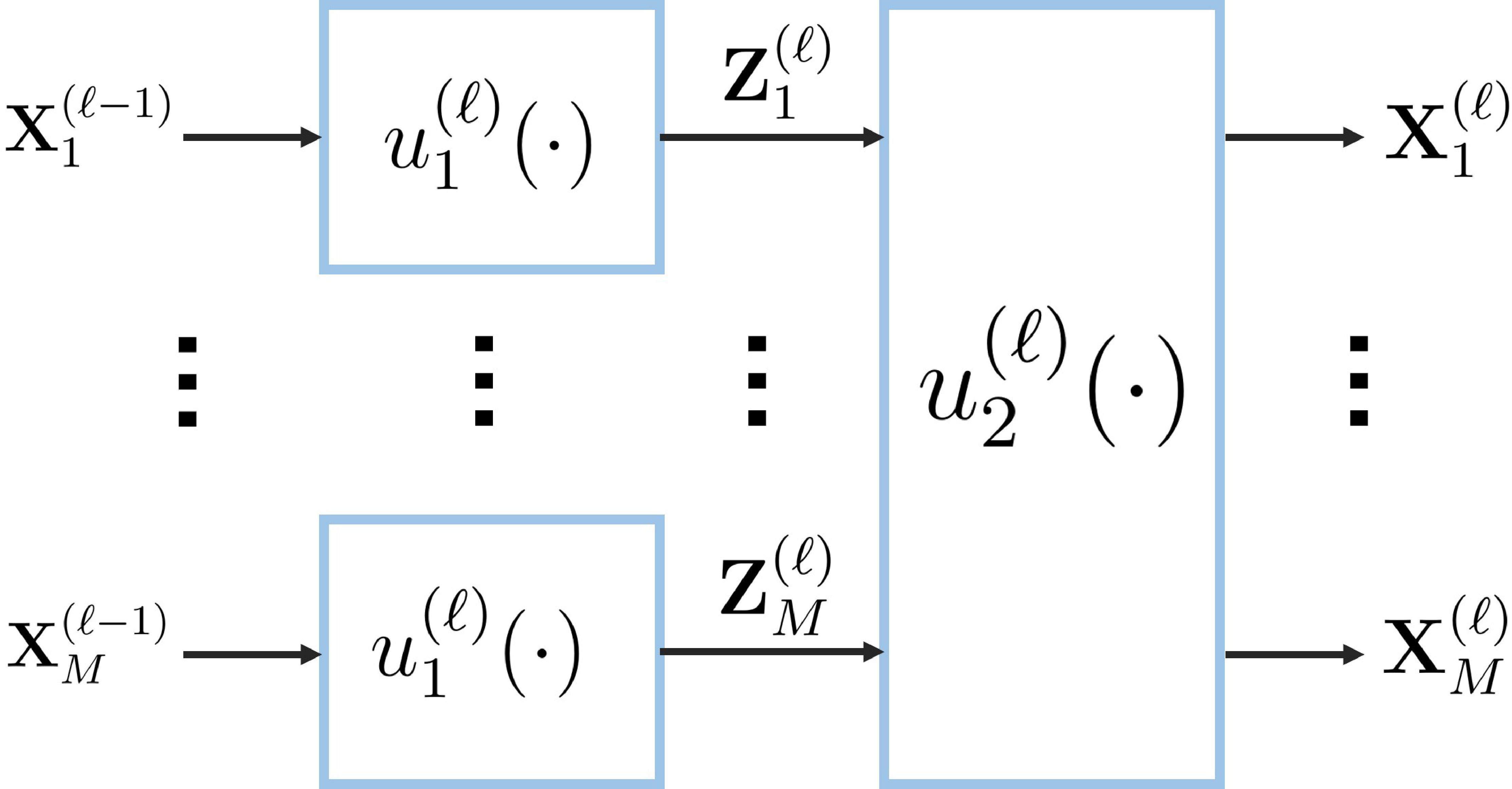}
  \caption{The $\ell$-th hierarchical layer of the proposed encoding block.}\label{HIE}
\end{center}
\vspace{-0.2cm}
\end{figure}

\vspace{2mm}
\emph{Remark 3:}
The hierarchical layers in Fig.~\ref{HIE} \textcolor{black}{are} the key to the HPE property. In particular,
this architecture satisfies the desired HPE property rather than the standard PE property,
making it indispensable to the multi-group multicast beamforming design\textcolor{black}{\footnote{
\textcolor{black}{If we embed the channel and beamforming weight over each antenna individually and let
$\mathbf{X}^{(\ell)}_{m,k}=\left[\mathbf{x}^{(\ell)}_{m,k,1},\mathbf{x}^{(\ell)}_{m,k,2},\ldots,\mathbf{x}^{(\ell)}_{m,k,N}\right]\in\mathbb{R}^{d\times N}$ denote the feature embedding over the $N$ antennas regarding the $k$-th user in group $m$,
we can add a self-attention block $u_0^{(\ell)}(\cdot)$
in front of $u_1^{(\ell)}(\cdot)$ to guarantee the PE property with respect to
the columns of $\mathbf{X}^{(\ell)}_{m,k}$,
which further provides the generalization ability on different numbers of antennas.}}}.
For example, if the input is permuted as
$\left[\tilde{\mathbf{x}}_{1,1}^{(\ell-1)}, \tilde{\mathbf{x}}_{1,2}^{(\ell-1)}, \tilde{\mathbf{x}}_{2,1}^{(\ell-1)}, \tilde{\mathbf{x}}_{2,2}^{(\ell-1)}\right]=\left[\mathbf{x}_{1,1}^{(\ell-1)}, \mathbf{x}_{2,1}^{(\ell-1)}, \mathbf{x}_{1,2}^{(\ell-1)}, \mathbf{x}_{2,2}^{(\ell-1)}\right]$, the output of the $\ell$-th hierarchical layer will not automatically satisfy
$\left[\tilde{\mathbf{x}}_{1,1}^{(\ell)}, \tilde{\mathbf{x}}_{1,2}^{(\ell)}, \tilde{\mathbf{x}}_{2,1}^{(\ell)}, \tilde{\mathbf{x}}_{2,2}^{(\ell)}\right]=\left[\mathbf{x}_{1,1}^{(\ell)}, \mathbf{x}_{2,1}^{(\ell)}, \mathbf{x}_{1,2}^{(\ell)}, \mathbf{x}_{2,2}^{(\ell)}\right]$.
Based on this architecture, the proposed encoding block can guarantee the HPE property in
\eqref{HPE1}, which will be formally proved later.

\subsubsection{De-Embedding Layer}
The de-embedding layer first transforms $\mathbf{X}^{(L)}\in\mathbb{R}^{d\times K}$
into $\mathbf{X}^{\text{de}}\in\mathbb{R}^{3\times K}$ by a linear projection:
\begin{eqnarray}\label{DE1}
\mathbf{X}^{\text{de}}=
\mathbf{W}^{\text{de}}\mathbf{X}^{(L)}+\mathbf{b}^{\text{de}}\otimes \mathbf{1}_{K}^\mathrm{T},
\end{eqnarray}
where $\mathbf{W}^{\text{de}}\in\mathbb{R}^{3\times d}$ and
$\mathbf{b}^{\text{de}}\in\mathbb{R}^3$
are the trainable matrix and vector, respectively.
It can be seen from $\eqref{DE1}$
that each column of $\mathbf{X}^{(L)}$ is processed independently and identically
in the de-embedding layer.
The corresponding output can be expanded as $\mathbf{X}^{\text{de}}=\left[\mathbf{X}^{\text{de}}_1,\mathbf{X}^{\text{de}}_2,\ldots,\mathbf{X}^{\text{de}}_M\right]$
and $\mathbf{X}^{\text{de}}_m=\left[\mathbf{x}^{\text{de}}_{m,1},\mathbf{x}^{\text{de}}_{m,2},\ldots,\mathbf{x}^{\text{de}}_{m,K_m}\right]$,
where $\mathbf{x}^{\text{de}}_{m,k}=\left[x^{\text{de}}_{m,k,1}, x^{\text{de}}_{m,k,2}, x^{\text{de}}_{m,k,3}\right]^\mathrm{T}$ contains the information
of $\left(\alpha_{m,k}, \lambda_{m,k}\right)$ by
\begin{eqnarray}\label{DE2}
\Re\left(\alpha_{m,k}\right)=x^{\text{de}}_{m,k,1},~~
\Im\left(\alpha_{m,k}\right)=x^{\text{de}}_{m,k,2},\nonumber\\
\lambda_{m,k} = \text{ReLU}\left(x^{\text{de}}_{m,k,3}\right),~~\forall~ k\in\mathcal{K}_m,~~\forall~ m\in\mathcal{M},
\end{eqnarray}
where $\text{ReLU}(\cdot)$ is adopted due to $\lambda_{m,k}\geq 0$.

Based on the above embedding layer, hierarchical layers, and de-embedding layer,
we establish the HPE property of the proposed encoding block in the
following theorem.

\newtheorem{theorem}{Theorem}
\begin{theorem}\label{HPE_EN}
(HPE Property of the Encoding Block)
The proposed encoding block for representing $f(\cdot)$ satisfies the HPE property in \eqref{HPE1}.
\end{theorem}
\begin{proof}
See Appendix \ref{proof1}.
\end{proof}
\vspace{2mm}

Theorem~\ref{HPE_EN} implies that the proposed encoding block
is inherently incorporated with the HPE property.
This is in sharp contrast to the PE property of the vanilla transformer in Property~\ref{PE},
where the PE property is even enforced with respect to any input components no matter
whether they are from the same group.

\subsection{Decoding Block for Representing $g(\cdot,\cdot,\cdot)$}
The proposed decoding block for representing $g(\cdot,\cdot,\cdot)$
consists of a solution construction layer and $R$ constraint augmented layers.
As shown in Fig.~~\ref{decoding}, the solution construction layer
utilizes $\left(\mathbf{H}, \boldsymbol{\alpha}, \boldsymbol{\lambda}\right)$
to construct an initial multicast beamforming matrix
$\mathbf{W}^{(0)}\in\mathbb{C}^{N\times M}$;
then, $\mathbf{W}^{(0)}$ is updated through
$R$ constraint augmented layers to obtain the final multicast
beamforming solution $\mathbf{W}=\mathbf{W}^{(R)}$.

\begin{figure}[t!]
\begin{center}
  \includegraphics[width=0.5\textwidth]{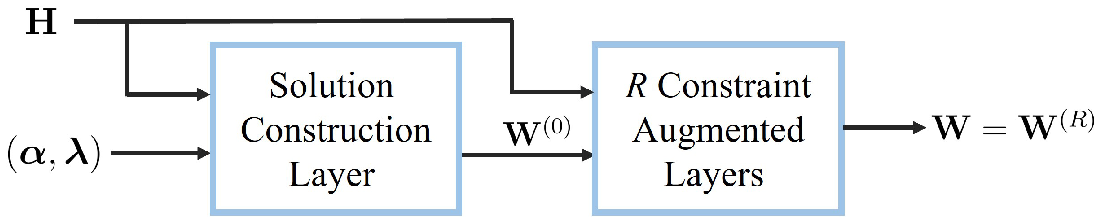}
  \caption{The architecture of the proposed decoding block for representing $g(\cdot,\cdot,\cdot)$.}\label{decoding}
\end{center}
\vspace{-0.2cm}
\end{figure}

\subsubsection{Solution Construction Layer}
The semi-closed form solution in \eqref{structure} provides an effective
guidance for establishing the mapping function from $\left(\mathbf{H}, \boldsymbol{\alpha}, \boldsymbol{\lambda}\right)$
to $\mathbf{W}$. Since the multicast beamforming structure is explicitly given by \eqref{structure},
the construction layer does not necessarily involve any trainable parameters.
In particular,
the solution construction layer computes an initial multicast beamforming $\mathbf{W}^{(0)}\triangleq\left[\mathbf{w}_{1}^{(0)},\mathbf{w}_{2}^{(0)},\ldots,\mathbf{w}_{M}^{(0)}\right]$
by a nonlinear operation:
\begin{eqnarray}\label{SC1}
\mathbf{w}_{m}^{(0)} = \left(\mathbf{I}_N+\sum_{j\in\mathcal{M}}
\sum_{k\in\mathcal{K}_j}
\lambda_{j,k}\gamma_{j}\mathbf{h}_{j,k}\mathbf{h}_{j,k}^\mathrm{H}\right)^{-1}\mathbf{H}_m\boldsymbol{\alpha}_m,
\nonumber\\
\forall~ m\in\mathcal{M}.
\end{eqnarray}
When the number of antennas $N$ is large, the computational complexity of the matrix inverse in \eqref{SC1} can be reduced by the
Woodbury matrix identity \cite{Higham}. Specifically, denoting $\mathbf{A}\triangleq\left[\mathbf{A}_1,\mathbf{A}_2,\ldots,\mathbf{A}_M\right]$,
$\mathbf{A}_m\triangleq\left[\mathbf{a}_{m,1},\mathbf{a}_{m,2}\ldots,\mathbf{a}_{m,K_m}\right]$, and $\mathbf{a}_{m,k}\triangleq \sqrt{\lambda_{m,k}\gamma_{m}}\mathbf{h}_{m,k}$,
the operation in \eqref{SC1} can be simplified as
\begin{equation}\label{SC2}
\mathbf{w}_{m}^{(0)} = \left(\mathbf{I}_N-
\mathbf{A}\left(\mathbf{I}_K+
\mathbf{A}^\mathrm{H}\mathbf{A}\right)^{-1}\mathbf{A}^\mathrm{H}\right)
\mathbf{H}_m\boldsymbol{\alpha}_m,
~~\forall~ m\in\mathcal{M}.
\end{equation}

While $\mathbf{W}^{(0)}$ may serve as a potential solution of
problem \eqref{MGMB} with
the corresponding $\mathbf{H}$, the nonconvex QoS constraints
in \eqref{MGMB_SINR} cannot be easily guaranteed without a dedicated layer design.
To better respect the constraints in \eqref{MGMB_SINR}, we further propose
$R$ constraint augmented layers to update the multicast beamforming matrix from $\mathbf{W}^{(0)}$ to $\mathbf{W}^{(R)}$.

\subsubsection{Constraint Augmented Layers}
\textcolor{black}{Moving the left-hand side of \eqref{MGMB_SINR} to the right-hand side, and noticing that
$x\leq 0$ is equivalent to $\text{ReLu}(x)=0$, we obtain
a more compact form of \eqref{MGMB_SINR}:}
\begin{eqnarray}\label{SINR2}
&&\textcolor{black}{{V_{m,k}\left(\mathbf{h}_{m,k},\mathbf{W}\right)}}\nonumber\\
&\triangleq&
\text{ReLU}\left(\gamma_m-\frac{\left\vert\mathbf{h}_{m,k}^{\mathrm{H}}\mathbf{w}_m\right\vert^2}
{\sum_{j=1,j\neq m}^{M} \left\vert\mathbf{h}_{m,k}^{\mathrm{H}}\mathbf{w}_j\right\vert^2+\sigma_{m,k}^2}\right)\nonumber\\
&=&0,~~\forall~ k\in \mathcal{K}_m,~~\forall~ m\in\mathcal{M}.
\end{eqnarray}
Then, the $K$ equality constraints in \eqref{SINR2} can be equivalently written as a single equation:
\begin{equation}\label{SINR3}
V\left(\mathbf{H},\mathbf{W}\right)\triangleq
\sum_{m\in\mathcal{M}}\sum_{k\in \mathcal{K}_m}V_{m,k}\left(\mathbf{h}_{m,k},\mathbf{W}\right)^2=0.
\end{equation}
Consequently, to guarantee the original constraints in \eqref{MGMB_SINR}, we alternatively
minimize $V\left(\mathbf{H},\mathbf{W}\right)$ over $\mathbf{W}$ to attain its minimum $0$.
To this end, we perform a number of gradient descent steps to decrease the function value of $V\left(\mathbf{H},\mathbf{W}\right)$.
In practice, if the initial multicast beamforming matrix $\mathbf{W}^{(0)}$ is close to the minimizer of
$V\left(\mathbf{H},\mathbf{W}\right)$, a few gradient descent steps are highly effective to attain the minimum of
$V\left(\mathbf{H},\mathbf{W}\right)$ \cite{Busseti19}.
These gradient descent steps are executed through $R$ constraint augmented layers.
As shown in Fig.~\ref{CA}, each constraint augmented layer performs a gradient descent update to the
multicast beamforming matrix:
\begin{eqnarray}\label{CAr}
\mathbf{W}^{(r)}  = \mathbf{W}^{(r-1)} -\eta \nabla_{\mathbf{W}}V\left(\mathbf{H},\mathbf{W}^{(r-1)}\right),\nonumber\\
\forall~ r\in\{1,2,\ldots,R\},
\end{eqnarray}
where $\eta>0$ denotes the step size,
and $\nabla_{\mathbf{W}}V\left(\mathbf{H},\mathbf{W}^{(r-1)}\right)$ denotes the gradient of
$V\left(\mathbf{H},\mathbf{W}\right)$ at $\mathbf{W}^{(r-1)}$.
\textcolor{black}{The function $V(\cdot,\cdot)$ given in \eqref{SINR3}
includes the SINR targets $\{\gamma_m\}_{m\in\mathcal{M}}$.
Consequently, once the SINR targets vary, the constraint augmented layers are adaptable
to the new SINR targets by adopting the new $\{\gamma_m\}_{m\in\mathcal{M}}$ in $V(\cdot,\cdot)$.}

\begin{figure}[t!]
\begin{center}
  \includegraphics[width=0.5\textwidth]{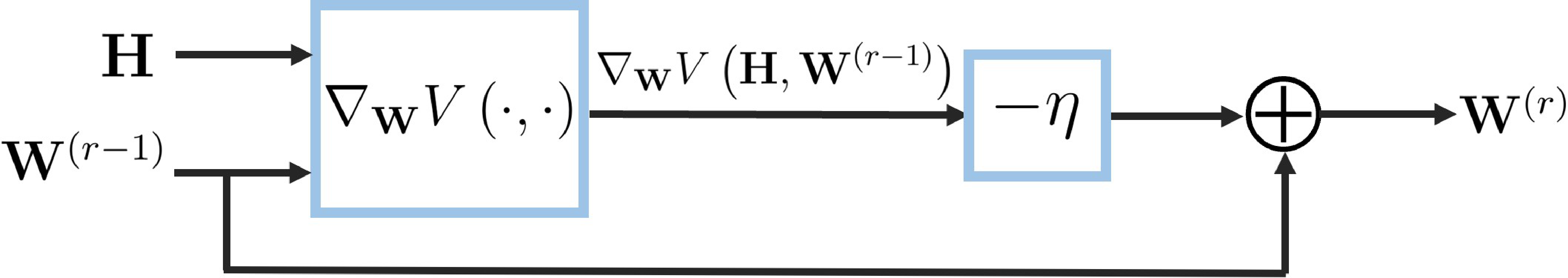}
  \caption{The $r$-th constraint augmented layer of the proposed decoding block.}\label{CA}
\end{center}
\vspace{-0.2cm}
\end{figure}

Notice that in \eqref{CAr}, the step size $\eta$ is the same in all of the $R$ constraint augmented layers. In this way,
we can set $R=R^{\text{train}}$ to be relatively small to facilitate
the back-propagation during the training procedure, while using a larger $R=R^{\text{test}}\geq R^{\text{train}}$
to ensure the feasibility of \eqref{MGMB_SINR}
during the test procedure.

Based on the above solution construction layer and the constraint augmented layers,
we establish the HPE property of the proposed decoding block in the following theorem.

\begin{theorem}\label{HPE_DE}
(HPE Property of the Decoding Block)
The proposed decoding block for representing $g(\cdot,\cdot,\cdot)$ satisfies the HPE property in \eqref{HPE2}.
\end{theorem}
\begin{proof}
See Appendix \ref{proof2}.
\end{proof}
\vspace{2mm}

\emph{Remark 4:} For the constraint augmented layers, a naive version
is to set $R^{\text{train}}=0$, while tuning $R^{\text{test}}$ to meet the QoS constraints.
In this case, since the HPE transformer is trained through $\mathbf{W}^{(0)}$,
the constraint augmented layers only play a role as a postprocessing step of the HPE transformer.
In contrast, when the HPE transformer is trained through $\mathbf{W}^{(R^{\text{train}})}$ with
$R^{\text{train}}>0$, the constraint augmented layers become a part of the HPE transformer, making it more powerful to respect the QoS constraints. In Section V, we will show the performance
gain of the proposed HPE transformer compared with the naive version of $R^{\text{train}}=0$.

\textcolor{black}{
\emph{Remark 5:}
The gradient descent steps are also adopted
to decrease the constraint violations in \cite{Donti21}.
In this sense, the proposed HPE transformer
can be viewed as an extension of \cite{Donti21} from generic fully connected neural networks
to the problem-dependent architecture with HPE property,
which empowers the constraint augmented layers
to well generalize on different problem scales.
In Section V, we will demonstrate through simulations that the constraint violations can be effectively mitigated
even if the problem scales are different from those of the training settings.}

\subsection{Training Procedure}
So far, we have presented the architecture of the proposed
HPE transformer for representing the
mapping function from $\mathbf{H}$ to $\mathbf{W}$, i.e., $\mathbf{W}=g\left(\mathbf{H},f\left(\mathbf{H}\right)\right)$.
With $\boldsymbol{\Theta}$ denoting all the trainable parameters in $g(\cdot,f(\cdot))$,
we add a penalty term into the loss function
to better balance the trade-off between the total transmit power and the violation of the QoS constraints.
In particular, the loss function during the training procedure is given by
\begin{equation}\label{loss}
\min_{\boldsymbol{\Theta}: \mathbf{W}=g\left(\mathbf{H},f\left(\mathbf{H}\right)\right)}~
\mathbb{E}\left[
\sum_{m=1}^M\left\Vert\mathbf{w}_{m}\right\Vert_2^2+\rho V\left(\mathbf{H}, \mathbf{W}\right)\right],
\end{equation}
where $\rho>0$ is a hyper-parameter for penalizing the constraint violation $V\left(\mathbf{H}, \mathbf{W}\right)$
defined in \eqref{SINR3},
and the expectation is taken over the channel matrix ${\mathbf{H}}$.

\begin{algorithm}[t!]
\caption{The Training Procedure for Optimizing $\boldsymbol{\Theta}$}
\begin{algorithmic}[1]\footnotesize
\State \textbf{Input:} number of epochs $N^{\text{e}}$,
steps per epoch $N^{\text{s}}$,
batch size $N^{\text{b}}$,
decay factor $\beta$,
and penalty hyper-parameter $\rho$;\\
\textbf{Initialize:} learning rate $\tau$;\\
$\textbf{for}$ $\text{epoch} = 1, 2,\ldots, N^{\text{e}}$\\
~~~~\textbf{for} $\text{step} = 1, 2,\ldots, N^{\text{s}}$\\
\begin{enumerate}[]
\item
~~~~a) Generate a batch of $N^{\text{b}}$ training samples $\left\{\mathbf{H}(n)\right\}_{n=1}^{N^{\text{b}}}$;
\item
~~~~b) Compute the mini-batch gradient of the loss function \eqref{loss} over $\boldsymbol{\Theta}$;
\item
~~~~c) Update $\boldsymbol{\Theta}$ by a stochastic gradient descent step using the Adam optimizer with learning rate $\tau$;
\end{enumerate}\\
~~~~\textbf{end}\\
~~~~$\tau\leftarrow\beta\tau$;\\
$\textbf{end}$\\
\textbf{Output:} HPE transformer $g(\cdot,f(\cdot))$ with optimized $\boldsymbol{\Theta}^{\star}.$
\end{algorithmic}
\end{algorithm}

Correspondingly, the training procedure for optimizing  $\boldsymbol{\Theta}$
is given in Algorithm~1. As shown in line 7, we adopt a learning rate decay strategy to
accelerate the training procedure \cite{You2019}. In particular, the learning rate $\tau$ is decreased by a factor $\beta$ at the end of each training epoch.
After the training procedure, the proposed
HPE Transformer $g(\cdot,f(\cdot))$ with the optimized trainable parameters $\boldsymbol{\Theta}^{\star}$
is used for multicast beamforming design given any channel matrix ${\mathbf{H}}$.

\section{Simulation Results}
In this section, we provide simulation results to demonstrate
the benefits of the proposed HPE transformer for solving the QoS constrained multi-group multicast beamforming design problem.


\subsection{Simulation Setting, Choices of Hyper-Parameters, and Benchmarks}
Consider a downlink multiuser system, where the $(x, y, z)$-coordinate of the BS in meters is $(0, 0, 20)$.
The users are uniformly distributed in a rectangular area $[85, 95]\times[85, 115]$ in the $(x,y)$-plane with $z=0$.
The large-scale fading coefficient is generated
according to the path-loss model $32.6 + 36.7 \log_{10} D$ in dB, where $D$ is the distance in meters between the BS and the user,
and the small-scale Rayleigh fading coefficient follows $\mathcal{CN}\left(0, 1\right)$.
The background Gaussian noise power is $-100$ dBm,
and the SINR target of each user is $10$ dB.

The hyper-parameters of the proposed HPE transformer are summarized as follows.
For the encoding block,
the embedding dimension is set as $d=128$,
the number of hierarchical layers is set as $L=2$,
the number of attention heads in the $\text{MHA}$ operation is $T=4$, and
the $\text{CFF}$ operation adopts a two-layer fully connected neural network with
$512$ hidden neurons using the ReLU activation.
For the decoding block, the number of constraint augmented layers is set as $R^{\text{train}}=5$,
and the step size of each constraint augmented layer
is set as $\eta=0.01$.

During the training procedure,
we apply the Adam optimizer in Pytorch
to optimize the trainable parameters of the proposed HPE transformer.
In Algorithm~1, the number of training epochs is set as $N^{\text{e}}=100$,
the number of gradient descent steps in each training epoch is
$N^{\text{s}}=2000$, and each step is updated using a batch of $N^{\text{b}}=1024$ training samples.
Unless otherwise specified, the learning rate is initialized as $\tau=10^{-4}$ and
the decay factor is set as $\beta=0.98$.

After the training procedure, the test performance
is evaluated over $N^{\text{t}}=1280$ test samples
in terms of both the total transmit power and the average constraint violation.
In particular, the average constraint violation is measured by
\begin{equation}\label{CV}
\text{CV}\triangleq
\frac{1}{K}
\sum_{m\in\mathcal{M}}\sum_{k\in \mathcal{K}_m}\frac{V_{m,k}\left(\mathbf{h}_{m,k},\mathbf{W}\right)}{\gamma_m},
\end{equation}
which reflects the relative violation to the SINR target.
\textcolor{black}{The $\text{CV}$ is a soft metric, which also considers the quality of those infeasible samples by measuring how much violation they cause. This metric is also consistent with the penalty term in the
loss function \eqref{loss}, which minimizes the violation of all the samples isotropically.}
All the experiments are implemented on the Tesla T4 GPU
and the Intel(R) Xeon(R) CPU @ 2.20GHz.

For comparison, we provide the simulation results of the following benchmarks.

\begin{itemize}
     \item \textbf{FC}: This is a 4-layer fully connected neural network, where
     the input is a $2KN$-dimensional real vector
     $\mathbf{h}^{\text{in}}=\text{vec}\left(\left[{\Re\left(\mathbf{H}\right)}^\mathrm{T},{\Im\left(\mathbf{H}\right)}^\mathrm{T}\right]^\mathrm{T}\right)$,
     and the output is a $2MN$-dimensional real vector
     $\mathbf{w}^{\text{out}}=\text{vec}\left(\left[{\Re\left(\mathbf{W}\right)}^\mathrm{T},{\Im\left(\mathbf{W}\right)}^\mathrm{T}\right]^\mathrm{T}\right)$.
     The number of neurons in each hidden layer is $\{4KN, 512, 4MN\}$, and
the activation function in each hidden layer is \text{ReLU}.

     \item \textbf{MMGNN}: A GNN based architecture was recently proposed for the max-min fairness
     multi-group multicast beamforming design \cite{Zhang22}.
     Since there is no peak power constraint in problem \eqref{MGMB}, the normalization of the output layer in MMGNN is omitted.
     Other than that, all the hyper-parameters of the neural network architecture are set according to Table~I in~\cite{Zhang22}.

     \item \textbf{Vanilla Transformer}: We adopt a vanilla transformer with
     $1$ embedding layer, $4$ self-attention blocks (see Fig.~\ref{AB}), and $1$ output layer.
     In particular, the output layer aggregates the hidden representations
     of the users in the same multicast group, and then utilizes a group-wise linear transformation to produce the corresponding beamformer. Other hyper-parameters are set the same as those of the proposed HPE transformer. \textcolor{black}{Notice that GNNs enjoy the PE property \cite{Peng23,Ziwei23}, while
     the vanilla transformer can be viewed as a special case of GNNs
     for fully connected graphs, and hence also enjoys the PE property.}

     \item \textbf{CCP}: This is an optimization-based iterative algorithm \cite{Yuille2003}, which is guaranteed to converge to a stationary point of problem \eqref{MGMB}. The initial point is obtained based on zero-forcing initialization \cite{Chen2017}. To ensure the convergence,
         the maximum iteration number is set as $10$ in the simulations.

          \item \textcolor{black}{\textbf{SCA-R}: This is a successive convex approximation algorithm
          to deal with the nonconvex rate constraints \cite{Li2015}, which is guaranteed to converge to a stationary point. To ensure the convergence,
         the maximum iteration number is set as $100$ in the simulations.}

    \item \textbf{ACR-BB}: This is a branch-and-bound based algorithm
     for single-group multicast beamforming design \cite{Lu2017}.
     Since ACR-BB is guaranteed to converge to the global optimal solution,
     we adopt it as a benchmark for the case of $M=1$ multicast group.
     The relative error tolerance
     is set as $10^{-4}$ in ACR-BB in order to obtain an almost global optimal solution.
\end{itemize}

\subsection{Performance Evaluation for Single Multicast Group}
\begin{figure}[t!]
\begin{center}
  \subfigure[Test performance on $K=4$.]{
  \includegraphics[width=0.35\textwidth]{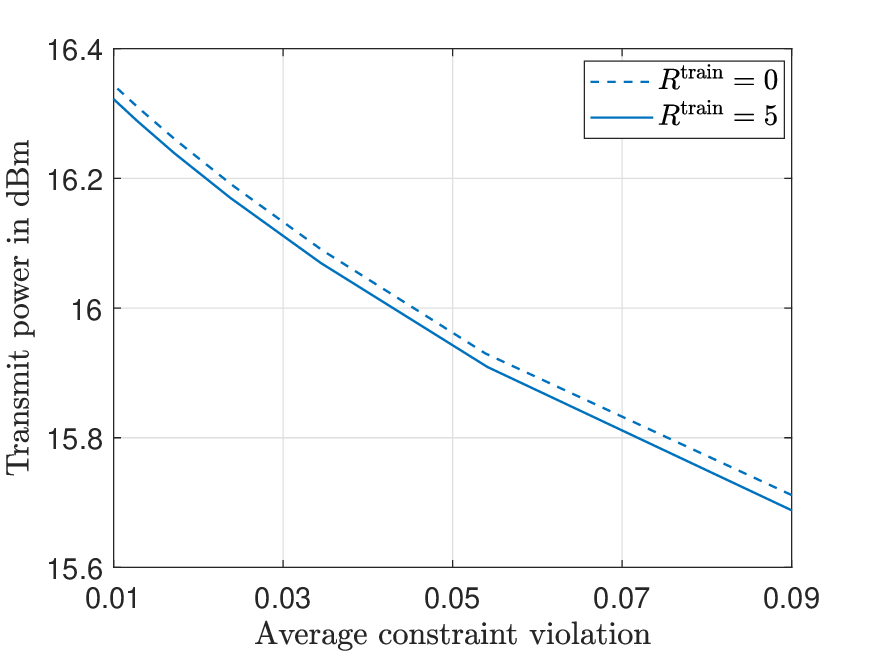}\label{power_CV_K4}
  }
 \subfigure[Test performance on $K=16$.]{
  \includegraphics[width=0.35\textwidth]{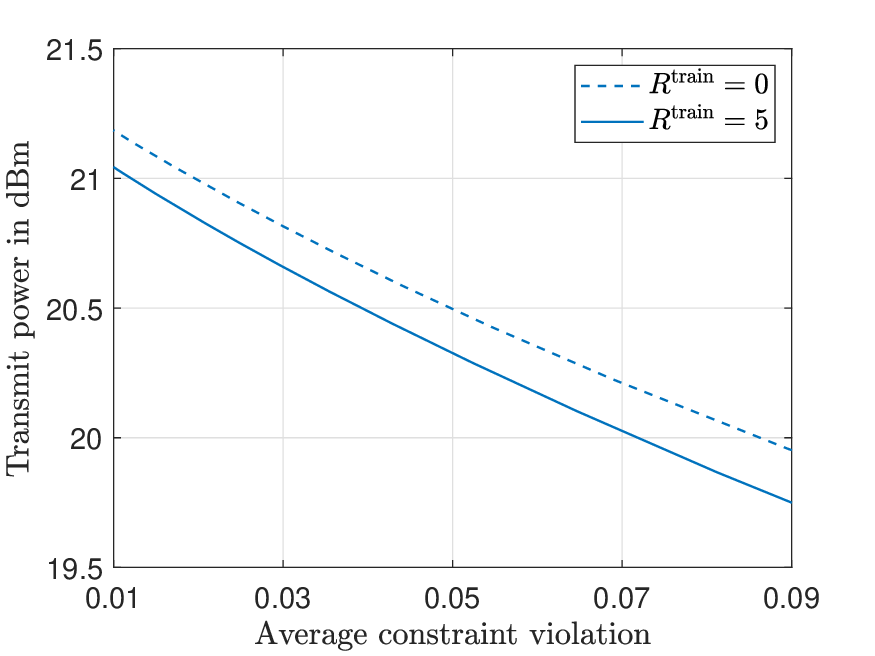}\label{power_CV_K16}
  }
\caption{Trade-off between the total transmit power and the average constraint violation.}\label{power_CV}
\end{center}
\vspace{-0.2cm}
\end{figure}

First, we demonstrate the performance of the proposed HPE transformer for
multicast beamforming design in a special case of $M=1$ multicast group. The number of antennas at the BS is set as $N=8$, and the number of users during the training procedure is fixed at $K=4$.
The penalty hyper-parameter of the HPE transformer is set as
$\rho=0.5$. We first
verify the effectiveness of the constraint augmented layers by comparing the
HPE transformer under $R^{\text{train}}=5$ with a naive version of $R^{\text{train}}=0$.
After training, we vary $R^{\text{test}}$ to realize a trade-off between the total transmit power and the average constraint violation as shown in Fig.~\ref{power_CV}.
The number of users during the test procedure is set as $K=4$ and $K=16$ in Fig.~\ref{power_CV_K4} and Fig.~\ref{power_CV_K16}, respectively.
It can been seen that $R^{\text{train}}=5$ achieves a
better trade-off between the transmit power and the constraint violation,
and the performance gain becomes larger when $K$ increases from $4$ to $16$.
In the rest of simulations, we fix $R^{\text{train}}=5$ and set
$R^{\text{test}}$ such that the average constraint violation is smaller than $0.01$.
\textcolor{black}{For fair comparison, the total transmit power is calculated based on the samples with $\text{CV}\leq0.05$ in the rest of simulations.
If an approach cannot even return any solution with $\text{CV}\leq0.05$, the corresponding total transmit power will be omitted.}

\begin{figure*}[t!]
\begin{center}
  \subfigure[Total transmit power versus $K$.]{
  \includegraphics[width=0.31\textwidth]{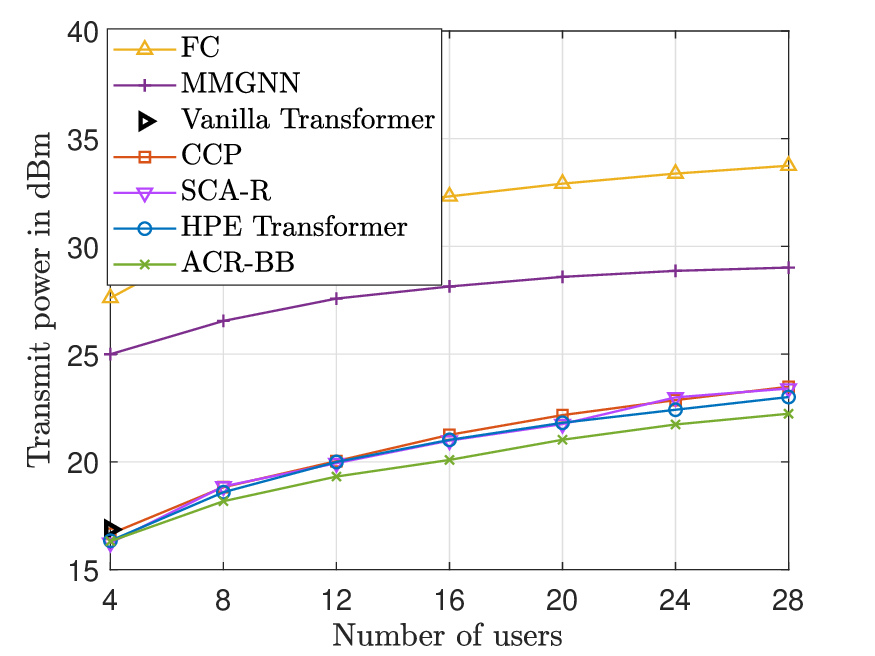}\label{power_M1}
  }
 \subfigure[Constraint violation versus $K$.]{
  \includegraphics[width=0.31\textwidth]{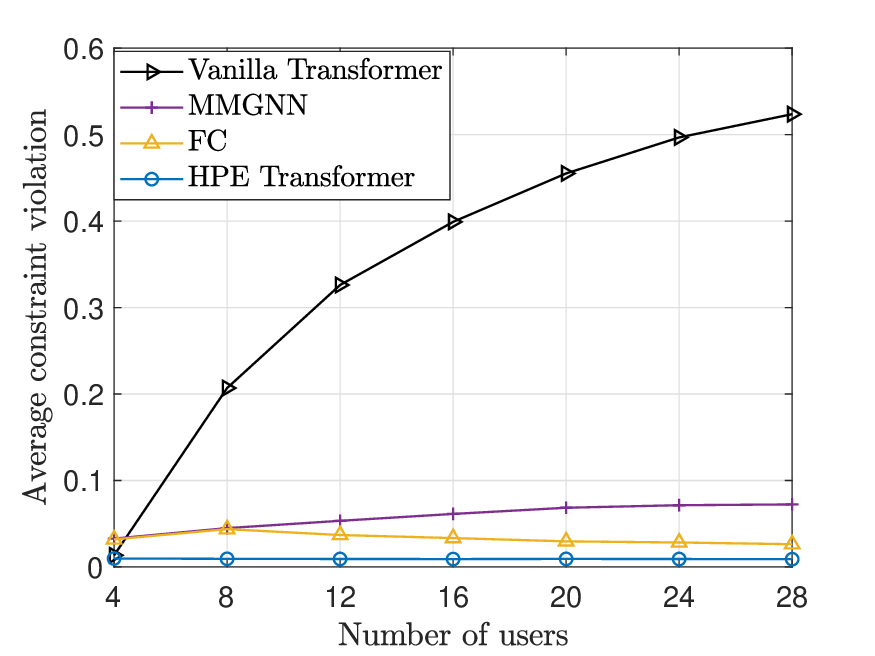}\label{CV_M1}
  }
  \subfigure[Computational time versus $K$.]{
  \includegraphics[width=0.31\textwidth]{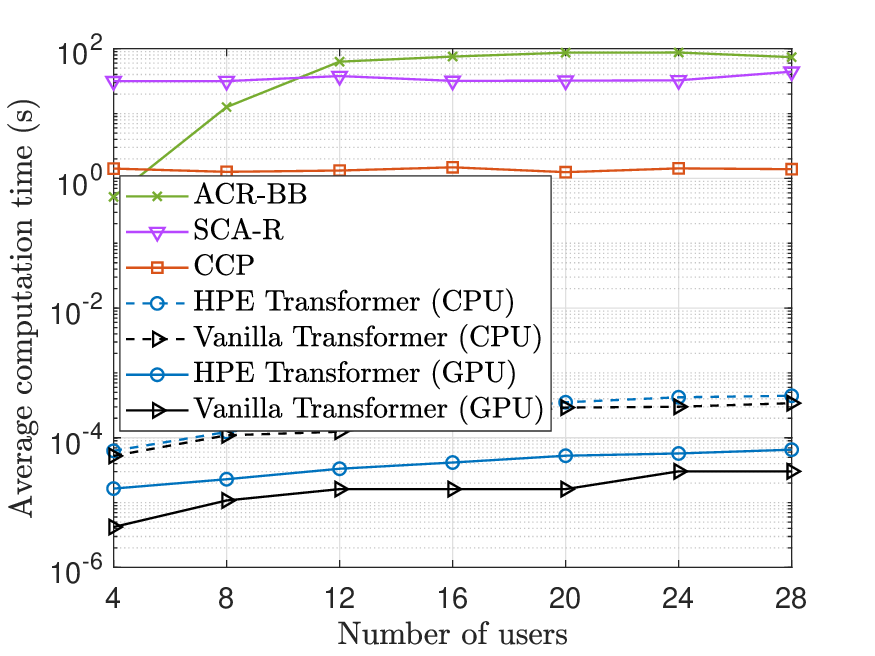}\label{T_M1}
}\caption{The generalization performance of the HPE transformer under different numbers of users for $M=1$ multicast group.}\label{M1}
\end{center}
\vspace{-0.2cm}
\end{figure*}

In Fig.~\ref{M1}, we illustrate the generalization performance of the proposed HPE transformer under different numbers of users from $4$ to $28$
and compare it with other benchmarks.
The penalty hyper-parameter of the vanilla transformer, MMGNN,
and FC is tuned as $\rho=5$, $20$, and $50$, respectively.
In Fig.~\ref{M1}, the proposed
HPE transformer, the vanilla transformer, and MMGNN are trained under $K=4$ while being tested under different $K$ from $4$ to $28$. Since the architecture of FC depends on the number of users,
it is retrained under different $K$ in Fig.~\ref{M1}.

From Fig.~\ref{power_M1}, we can see that as $K$ increases, the proposed
HPE transformer always achieves much lower transmit power than FC and MMGNN,
and the transmit power of the proposed HPE transformer
is very close to that of the CCP, \textcolor{black}{SCA-R,}
and ACR-BB.
\textcolor{black}{Note that ACR-BB has the lowest transmit power because it is a branch-and-bound based algorithm, which is guaranteed to find the global optimal solution.}
Moreover, Fig.~\ref{CV_M1} demonstrates that the superiority of the proposed
HPE transformer in the transmit power is achieved without sacrificing the QoS constraints.
In contrast, the other three deep learning-based methods
cannot well satisfy the QoS constraints.
To demonstrate the superiority of the proposed HPE transformer in real-time implementations,
we further show its computational time on GPU/CPU in Fig.~\ref{T_M1},
where the computational time of \textcolor{black}{FC and MMGNN}
is omitted due to their much higher transmit power and/or much larger constraint violation.
It can be seen that the computational time of the proposed HPE transformer is much less than that of the optimization-based CCP, \textcolor{black}{SCA-R,} and ACR-BB.
\textcolor{black}{Since SCA-R achieves a similar total transmit power
as that of CCP, but requires a much longer computational time,
we omit its results in the rest of simulations.}

\subsection{Performance Evaluation for Multiple Multicast Groups}
Next, we compare the performance of different approaches for $M=3$
multicast groups.
The number of antennas at the BS is set as $N=16$,
and the number of users in each group $m$
during the training procedure is fixed at $K_m=4$,
i.e., the total number of users is $K=12$.
In Fig.~\ref{M3}, we show the generalization performance of the proposed HPE transformer
as $K$ increases from $3$ to $15$, where
each multicast group contains the same number of users.
The learning rate is initialized as $\tau=10^{-3}$ and
the decay factor is set as $\beta=0.96$.
The penalty hyper-parameter of the proposed HPE transformer, the vanilla transformer, and MMGNN is
tuned as $\rho=0.2$, $10$, and $10$, respectively.
Since FC performs the worst and requires to be retrained under different $K$,
we omit it in the rest of simulations.

\begin{figure*}[t!]
\begin{center}
  \subfigure[Total transmit power versus $K$.]{
  \includegraphics[width=0.31\textwidth]{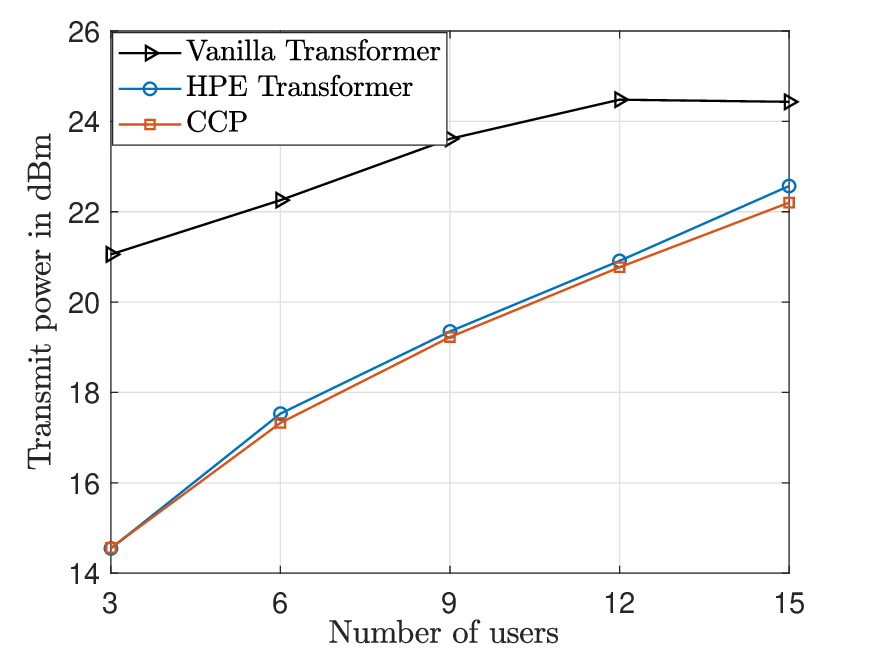}\label{power_M3}
  }
 \subfigure[Constraint violation versus $K$.]{
  \includegraphics[width=0.31\textwidth]{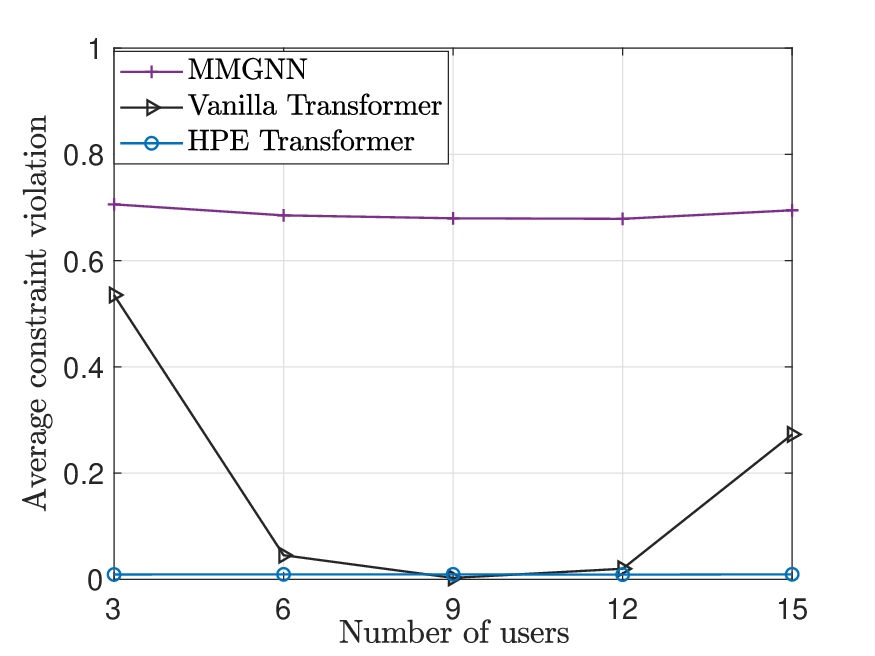}\label{CV_M3}
  }
  \subfigure[Computational time versus $K$.]{
  \includegraphics[width=0.31\textwidth]{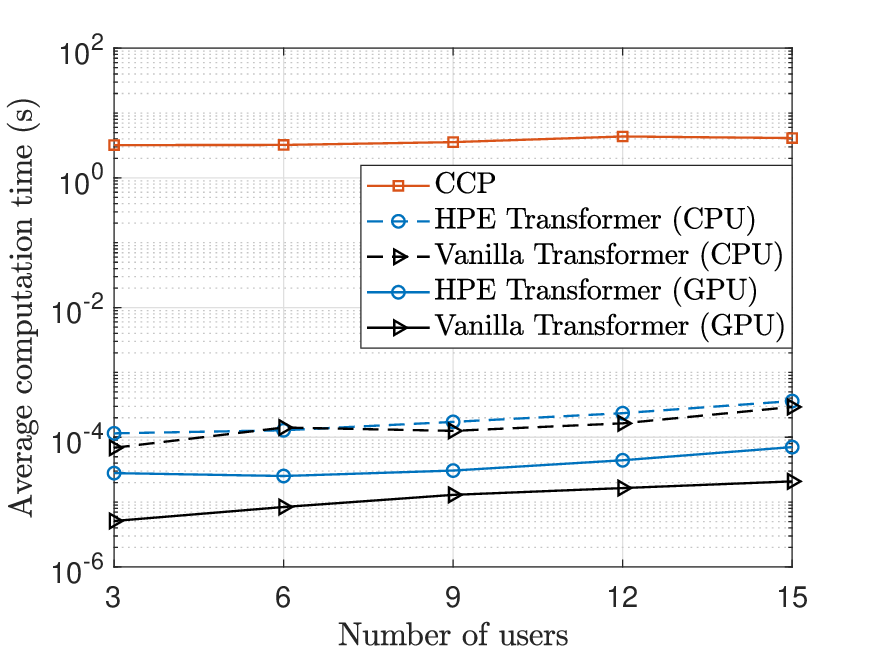}\label{T_M3}
}\caption{The generalization performance of the HPE transformer
under different numbers of users for $M=3$ multicast groups.}\label{M3}
\end{center}
\vspace{-0.2cm}
\end{figure*}

From Fig.~\ref{power_M3}, we can see that the transmit power of the proposed
HPE transformer generalized on different $K$ is nearly the same as
that of the optimization-based CCP.
In contrast, from Fig.~\ref{power_M3} and Fig.~\ref{CV_M3},
we can see that the vanilla transformer performs pretty badly especially when $K$ is different from that in the training procedure. Besides, MMGNN performs
\textcolor{black}{the worst}
even when $K$ is the same as that in the training procedure.
Moreover, Fig.~\ref{T_M3} shows that the computational speed of the proposed
HPE transformer is $10^4\sim10^5$ times faster than the optimization-based CCP
under different $K$.

\begin{figure*}[t!]
\begin{center}
  \subfigure[Total transmit power versus $M$.]{
  \includegraphics[width=0.31\textwidth]{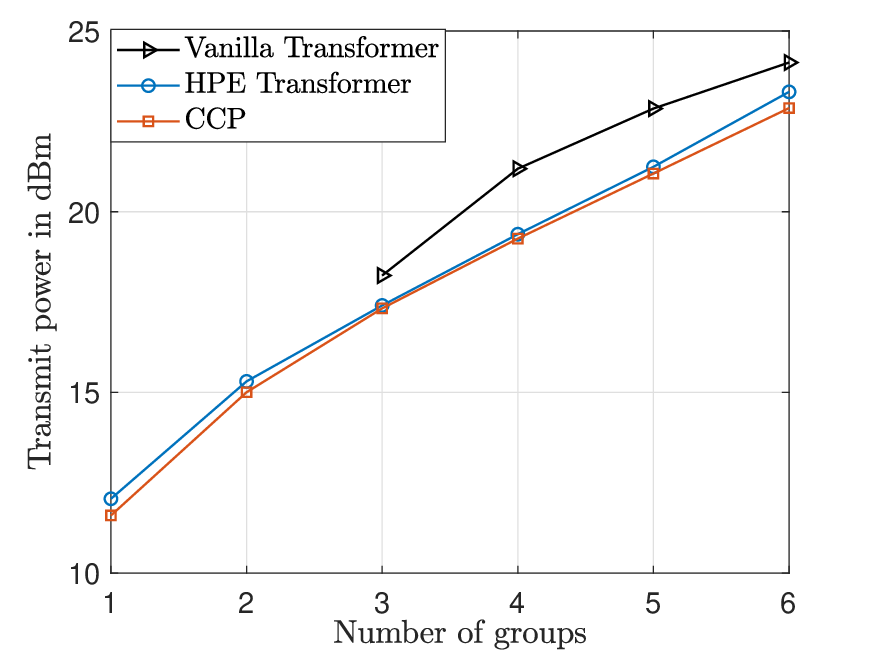}\label{power_M4}
  }
 \subfigure[Constraint violation versus $M$.]{
  \includegraphics[width=0.31\textwidth]{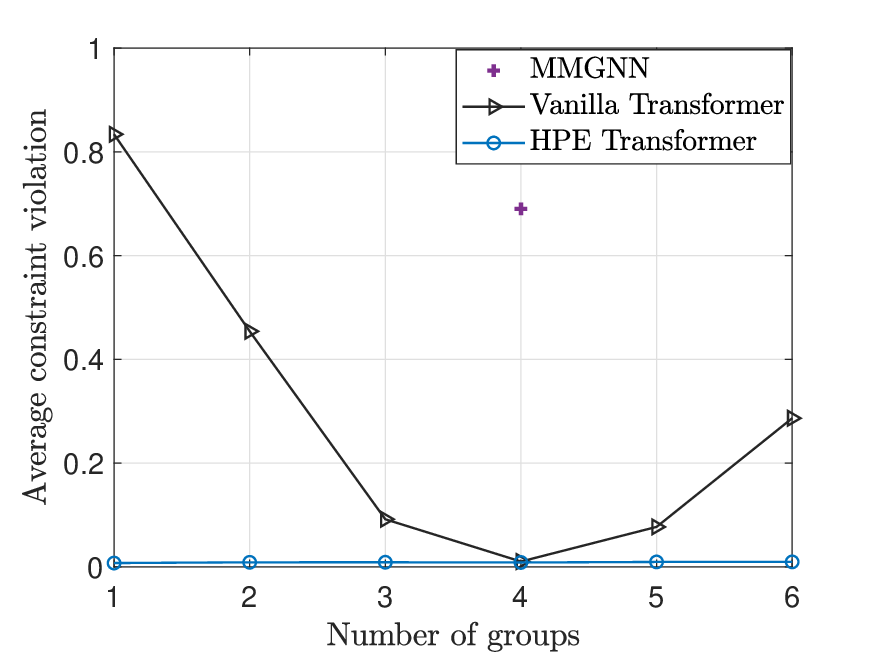}\label{CV_M4}
  }
  \subfigure[Computational time versus $M$.]{
  \includegraphics[width=0.31\textwidth]{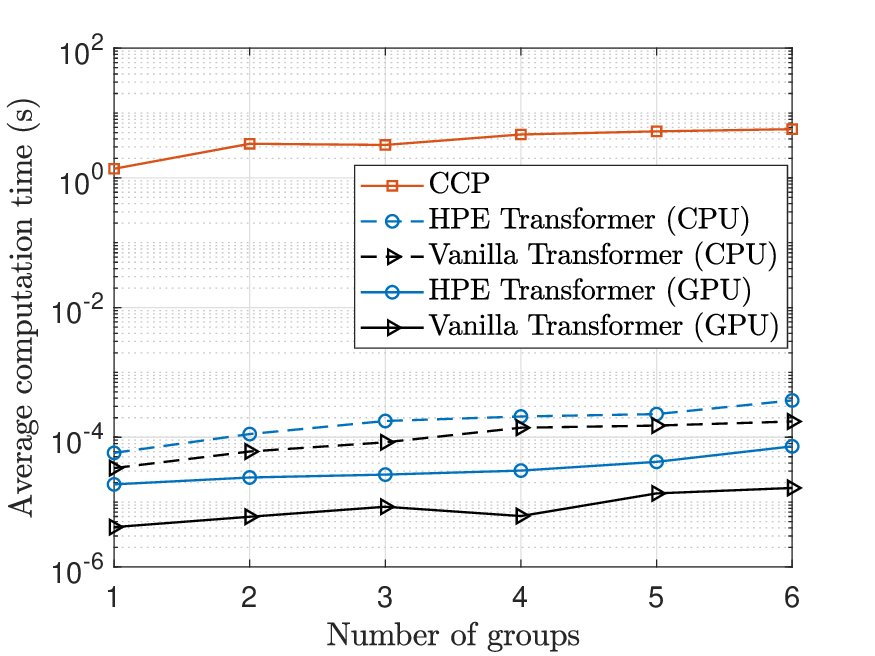}\label{T_M4}
}\caption{The generalization performance of the HPE transformer under different multicast groups.}\label{M4}
\end{center}
\vspace{-0.2cm}
\end{figure*}

In addition to the generalization performance on different numbers of users,
we further demonstrate the generalization performance of the proposed HPE transformer under different
numbers of multicast groups.
We fix the number of multicast groups during the training procedure
at $M=4$, while testing the generalization performance under different $M$
from $1$ to $6$ in Fig.~\ref{M4}.
The number of antennas at the BS is set as $N=16$,
and the number of users in each multicast group is $K_m=2$.
The penalty hyper-parameter of the proposed HPE transformer, the vanilla transformer, and MMGNN is tuned as
$\rho=0.2$, $10$, and $10$, respectively.
Since a well-trained MMGNN cannot be used for a different
number of multicast groups, we only show the simulation results of MMGNN when $M=4$.

It can be seen from Fig.~\ref{power_M4} that the transmit power of
the proposed HPE transformer generalized on different $M$ is very close to
that of the optimization-based CCP.
In contrast, as shown in
Fig.~\ref{power_M4} and Fig.~\ref{CV_M4}, the performance of the vanilla transformer is much worse
(compared with that of the proposed HPE transformer)
especially when $M$ is different from that in the training procedure.
Although MMGNN is tested under the same
number of multicast groups as that in the training procedure,
its constraint violation \textcolor{black}{is} the largest.
Moreover, Fig.~\ref{T_M4} shows that the proposed HPE transformer achieves remarkably faster computational speed than the optimization-based CCP under different $M$.

\begin{figure*}[t!]
\begin{center}
  \subfigure[Total transmit power versus SINR target.]{
  \includegraphics[width=0.31\textwidth]{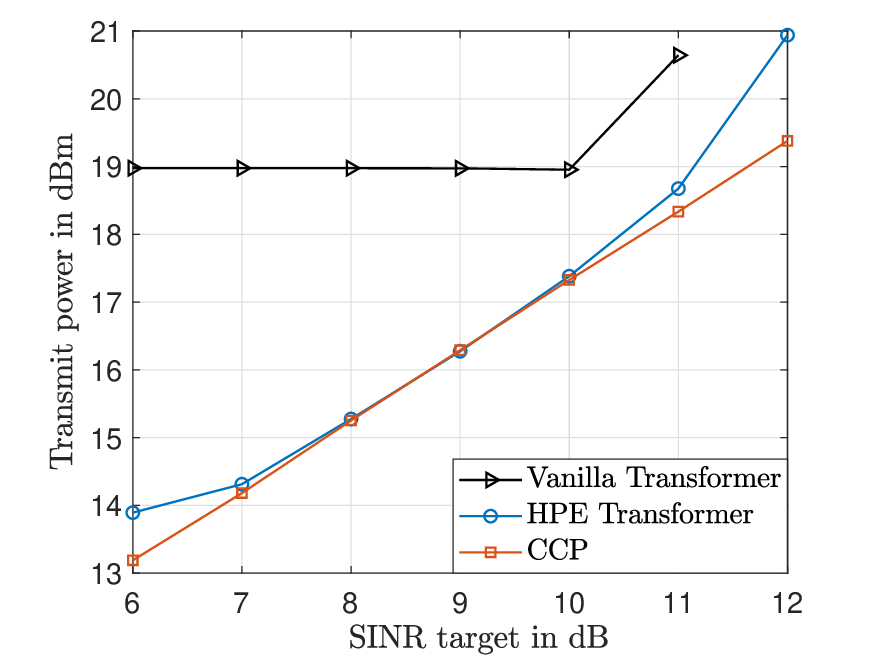}\label{power_gamma}
  }
 \subfigure[Constraint violation versus SINR target.]{
  \includegraphics[width=0.31\textwidth]{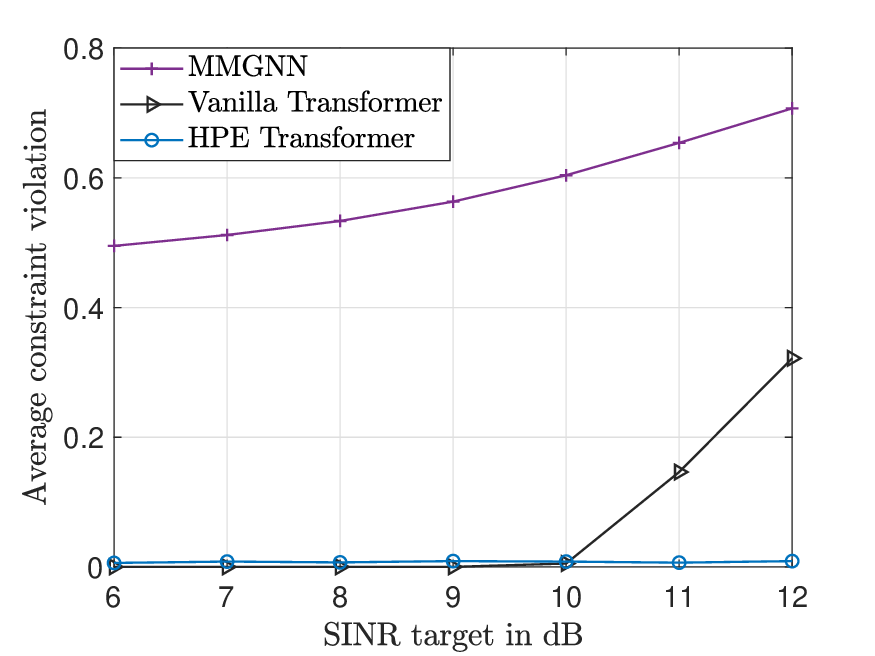}\label{CV_gamma}
  }
  \subfigure[Computational time versus SINR target.]{
  \includegraphics[width=0.31\textwidth]{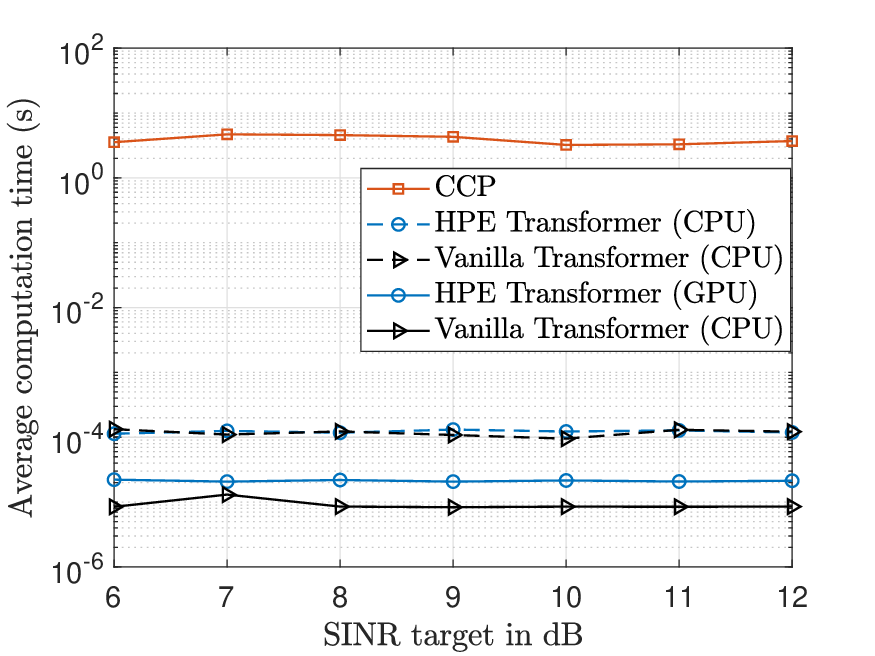}\label{T_gamma}
}
\caption{The generalization performance of the HPE transformer under different SINR targets.}\label{gamma}
\end{center}
\vspace{-0.2cm}
\end{figure*}

Finally, we demonstrate the generalization performance of the proposed HPE transformer under different
SINR targets. During the training procedure, we fix the SINR target at $10$ dB, while
we vary the SINR target in the decoding block from $6$ to $12$ dB during the test procedure
to verify the generalization performance under different SINR targets.
The number of antennas at the BS is set as $N=16$,
the number of multicast groups is $M=3$,
and the number of users in each multicast group is $K_m=2$.
The penalty hyper-parameter of the proposed HPE transformer, the vanilla transformer, and MMGNN is tuned as
$\rho=0.2$, $10$, and $10$, respectively.

We can see from Fig.~\ref{power_gamma} that
the proposed HPE transformer generalizes well on different SINR targets in terms of
the total transmit power, while the vanilla transformer \textcolor{black}{results} in a large performance gap.
Furthermore, as shown in Fig.~\ref{CV_gamma},
the constraint violation of the proposed HPE transformer is always very small under different SINR targets.
In addition to the superior generalization performance under different SINR targets,
Fig.~\ref{T_gamma} further shows that such a
generalization ability will not sacrifice the computational time, which is in sharp contrast to the
optimization-based CCP.

\section{Conclusions}
This paper proposed a deep learning-based approach for the QoS constrained multi-group multicast beamforming \textcolor{black}{design}.
By judiciously exploiting the multicast beamforming structure, the mapping function from wireless channels
to multicast beamformers was effectively decomposed into two separate mapping functions, both of which
inherently have an HPE property. To incorporate the HPE property into the neural network architecture design,
we proposed an HPE transformer with an encoding block and a decoding block for representing the corresponding two mapping functions, respectively.
We also proved that both of the blocks can guarantee the desired HPE property.
Simulation results showed that the proposed HPE transformer not only achieves
much lower transmit power and much smaller constraint violation than state-of-the-art deep learning-based approaches
for multicast beamforming design, but also takes remarkably less computational time than the optimization-based iterative algorithms.
Moreover, the proposed HPE transformer was demonstrated to generalize pretty well on different numbers of users, different numbers of multicast groups, and different SINR targets due to the HPE property.

\numberwithin{equation}{section}
\appendices

\section{The operation of $\text{MHA}$}\label{MHA_expression}
The $\text{MHA}$ operation uses $T$ attention heads to
extract the relevance among the input components. Let $\mathbf{W}^{\text{q}}_t\in\mathbb{R}^{d^{\prime}\times d}$,
$\mathbf{W}^{\text{k}}_t\in\mathbb{R}^{d^{\prime}\times d}$,
and $\mathbf{W}^{\text{v}}_t\in\mathbb{R}^{d^{\prime}\times d}$
denote three trainable matrices,
where $d^{\prime}=d/T$ and $t\in\mathcal{T}\triangleq\{1,2,\ldots,T\}$.
The $t$-th attention head computes a query, a key, and a value,
which are respectively given by
\begin{equation}\label{qkv}
\mathbf{Q}_{t}=\mathbf{W}^{\text{q}}_t\mathbf{X},~~
\mathbf{K}_{t}=\mathbf{W}^{\text{k}}_t\mathbf{X},~~
\mathbf{V}_{t}=\mathbf{W}^{\text{v}}_t\mathbf{X},~~
~~\forall~ t\in\mathcal{T}.
\end{equation}
To measure how similar the query and the key is,
each attention head $t$ computes a dot product:
\begin{equation}\label{compatibility}
\mathbf{A}_t=\frac{\mathbf{K}_{t}^\mathrm{T}\mathbf{Q}_{t}}{\sqrt{d^{\prime}}},
~~\forall~ t\in\mathcal{T}.
\end{equation}
Then, $\mathbf{A}_t$ is normalized such that the sum of each column is equal to $1$,
which is given by
\begin{equation}\label{weight}
\dot{a}_{i,j,t}=\frac{e^{a_{i,j,t}}}{\sum_{i^{\prime}=1}^I e^{a_{i^{\prime},j,t}}},~~\forall~ i\in\mathcal{I},
~~\forall~ j\in\mathcal{I},
~~\forall~ t\in\mathcal{T},
\end{equation}
where $a_{i,j,t}$ and $\dot{a}_{i,j,t}$
denote the $(i,j)$-th entry in $\mathbf{A}_t$ and $\dot{\mathbf{A}}_t$, respectively.
With $\dot{a}_{i,j,t}$ scoring the contribution of
$\mathbf{x}_i$ to $\mathbf{x}_j$, the attention value at the $t$-th attention head
is computed as
\begin{equation}\label{attention}
\mathbf{X}^{\prime}_{t}= \mathbf{V}_{t}\dot{\mathbf{A}}_{t},~~\forall~ t\in\mathcal{T}.
\end{equation}
Finally, by combining the attention values at the $T$ attention heads,
we obtain the computational result of $\text{MHA}$:
\begin{equation}\label{MHAvalue}
\text{MHA}\left(\mathbf{X}\right)=
\sum_{t=1}^T \mathbf{W}^{\text{o}}_t\mathbf{X}^{\prime}_{t},
\end{equation}
where $\mathbf{W}^{\text{o}}_t\in\mathbb{R}^{d\times d^{\prime}}$ is a trainable matrix.

\section{Proof of Theorem \ref{HPE_EN}}\label{proof1}
Let the mapping function from $\mathbf{H}$ to $\mathbf{X}^{(0)}$
given by \eqref{EM1} and \eqref{EM2} as $f^{\text{em}}(\cdot)$.
Define $\tilde{\mathbf{X}}^{(\ell)}\triangleq\left[
\tilde{\mathbf{X}}^{(\ell)}_{\pi_0(1)},\tilde{\mathbf{X}}^{(\ell)}_{\pi_0(2)},\ldots,\tilde{\mathbf{X}}^{(\ell)}_{\pi_0(M)}\right]$
and $\tilde{\mathbf{X}}^{(\ell)}_m\triangleq\left[
\mathbf{x}^{(\ell)}_{m,\pi_m(1)},\mathbf{x}^{(\ell)}_{m,\pi_m(2)},\ldots,\mathbf{x}^{(\ell)}_{m,\pi_m(K_m)}\right]$.
Since $f^{\text{em}}(\cdot)$ processes each column of $\mathbf{H}$ independently and identically,
we have
\begin{eqnarray}\label{C1}
\tilde{\mathbf{X}}^{(0)}=f^{\text{em}}\left(\tilde{\mathbf{H}}\right),
~~\forall~ \pi_{m}(\cdot):\mathcal{K}_m\rightarrow \mathcal{K}_m,
\nonumber\\\forall~ m\in\mathcal{M},~~
\forall~ \pi_{0}(\cdot):\mathcal{M}\rightarrow \mathcal{M}.
\end{eqnarray}
Furthermore, let the mapping function from $\mathbf{X}^{(\ell-1)}$ to $\mathbf{X}^{(\ell)}$
given by \eqref{H1} and \eqref{H2} as $f^{(\ell)}(\cdot)$ for all $\ell\in\mathcal{L}$.
Accordingly, we have
\begin{equation}\label{C2}
\mathbf{X}^{(\ell)}=f^{(\ell)}\left(\mathbf{X}^{(\ell-1)}\right)=
u_2^{(\ell)}\left(\textcolor{black}{\mathbf{Z}
^{(\ell)}}\right),~~\forall~ \ell\in\mathcal{L},
\end{equation}
where \textcolor{black}{$\mathbf{Z}^{(\ell)}\triangleq\left[
\mathbf{Z}_1^{(\ell)},\mathbf{Z}_2^{(\ell)},\ldots,\mathbf{Z}_M^{(\ell)}\right]$
and $\mathbf{Z}_m^{(\ell)}\triangleq
\left[\mathbf{z}_{m,1}^{(\ell)}, \mathbf{z}_{m,2}^{(\ell)}, \ldots, \mathbf{z}_{m,K_m}^{(\ell)}\right]
\triangleq u_1^{(\ell)}\left(\mathbf{X}^{(\ell-1)}_m\right)$.
Let $\tilde{\mathbf{Z}}^{(\ell)}\triangleq\left[
\tilde{\mathbf{Z}}^{(\ell)}_{\pi_0(1)},\tilde{\mathbf{Z}}^{(\ell)}_{\pi_0(2)},\ldots,\tilde{\mathbf{Z}}^{(\ell)}_{\pi_0(M)}\right]$
and $\tilde{\mathbf{Z}}^{(\ell)}_m
\triangleq\left[
\mathbf{z}^{(\ell)}_{m,\pi_m(1)},\mathbf{z}^{(\ell)}_{m,\pi_m(2)},\ldots,\mathbf{z}^{(\ell)}_{m,\pi_m(K_m)}\right]$.
Since $u_1^{(\ell)}(\cdot)$ satisfies the PE property in \eqref{PEu},
we have $\tilde{\mathbf{Z}}^{(\ell)}_m
\triangleq u_1^{(\ell)}\left(\tilde{\mathbf{X}}^{(\ell-1)}_m\right)$
for all $\pi_{m}(\cdot):\mathcal{K}_m\rightarrow \mathcal{K}_m$.
Substituting this equation into \eqref{C2}, we have
\begin{eqnarray}\label{C3}
f^{(\ell)}\left(\tilde{\mathbf{X}}^{(\ell-1)}\right)=
u_2^{(\ell)}\left(\tilde{\mathbf{Z}}^{(\ell)}
\right),~~\forall~ \ell\in \mathcal{L},
\nonumber\\
\forall~ \pi_{m}(\cdot):\mathcal{K}_m\rightarrow \mathcal{K}_m,~~\forall~ m\in\mathcal{M},~~
\forall~ \pi_{0}(\cdot):\mathcal{M}\rightarrow \mathcal{M}.
\end{eqnarray}
Furthermore, since $u_2^{(\ell)}(\cdot)$ also satisfies the PE property in \eqref{PEu},
combining \eqref{C2} and \eqref{C3} yields
\begin{eqnarray}\label{C4}
f^{(\ell)}\left(\tilde{\mathbf{X}}^{(\ell-1)}\right)=
u_2^{(\ell)}\left(\tilde{\mathbf{Z}}^{(\ell)}
\right)=\tilde{\mathbf{X}}^{(\ell)},~~\forall~ \ell\in \mathcal{L},\nonumber\\
\forall~ \pi_{m}(\cdot):\mathcal{K}_m\rightarrow \mathcal{K}_m,~~\forall~ m\in\mathcal{M},~~
\forall~ \pi_{0}(\cdot):\mathcal{M}\rightarrow \mathcal{M}.
\end{eqnarray}}
Finally, let the mapping function from $\mathbf{X}^{(L)}$ to $\left(\boldsymbol{\alpha}, \boldsymbol{\lambda}\right)$
given by \eqref{DE1} and \eqref{DE2} as $f^{\text{de}}(\cdot)$.
Since $f^{\text{de}}(\cdot)$ processes each column of $\mathbf{X}^{(L)}$
independently and identically, we have
\begin{eqnarray}\label{C5}
\left(\tilde{\boldsymbol{\alpha}}, \tilde{\boldsymbol{\lambda}}\right)=f^{\text{de}}\left(\tilde{\mathbf{X}}^{(L)}\right),
~~\forall~ \pi_{m}(\cdot):\mathcal{K}_m\rightarrow \mathcal{K}_m,\nonumber\\
\forall~ m\in\mathcal{M},~~\forall~ \pi_{0}(\cdot):\mathcal{M}\rightarrow \mathcal{M}.
\end{eqnarray}
Combining \eqref{C1}, \eqref{C4}, and \eqref{C5} yields the HPE property given in
\eqref{HPE1}.

\section{Proof of Theorem \ref{HPE_DE}}\label{proof2}
Let the mapping function from $\left(\mathbf{H}, \boldsymbol{\alpha}, \boldsymbol{\lambda}\right)$ to $\mathbf{W}^{(0)}$
given by \eqref{SC1} (or equivalently \eqref{SC2}) as $g^{\text{sc}}(\cdot,\cdot,\cdot)$.
Substituting $\left(\tilde{\mathbf{H}}, \tilde{\boldsymbol{\alpha}}, \tilde{\boldsymbol{\lambda}}\right)$ into the right-hand side of
\eqref{SC1}, we have \eqref{D1} on the next page,
\begin{figure*}[!t]
\normalsize
\begin{eqnarray}\label{D1}
 && \left(\mathbf{I}_N+\sum_{j\in\mathcal{M}}
\sum_{k\in\mathcal{K}_{\pi_0(j)}}
\lambda_{\pi_0(j),\pi_k(k)}\gamma_{\pi_0(j)}\mathbf{h}_{\pi_0(j),\pi_k(k)}\mathbf{h}_{\pi_0(j),\pi_k(k)}^\mathrm{H}\right)^{-1}\tilde{\mathbf{H}}_{\pi_0(m)}\tilde{\boldsymbol{\alpha}}_{\pi_0(m)}
\nonumber\\
&=&
\left(\mathbf{I}_N+\sum_{j\in\mathcal{M}}
\sum_{k\in\mathcal{K}_j}
\lambda_{j,k}\gamma_{j}\mathbf{h}_{j,k}\mathbf{h}_{j,k}^\mathrm{H}\right)^{-1}\mathbf{H}_{\pi_0(m)}\boldsymbol{\alpha}_{\pi_0(m)}=\mathbf{w}_{\pi_0(m)}^{(0)},
~~\forall~ m\in\mathcal{M},
\end{eqnarray}
\hrulefill
\vspace*{4pt}
\end{figure*}
where the first equality holds due to the permutation invariance of the summation, and the second equality holds due to \eqref{SC1}.
Define $\tilde{\mathbf{W}}^{(r)}\triangleq\left[
\mathbf{w}^{(r)}_{\pi_0(1)},\mathbf{w}^{(r)}_{\pi_0(2)},
\ldots,\mathbf{w}^{(r)}_{\pi_0(M)}\right], \forall~ r\in\{0,1,\ldots,R\}$.
Rewriting \eqref{D1} into a compact form, we have
\begin{eqnarray}\label{D2}
\tilde{\mathbf{W}}^{(0)}=g^{\text{sc}}\left(\tilde{\mathbf{H}}, \tilde{\boldsymbol{\alpha}}, \tilde{\boldsymbol{\lambda}}\right),~~
\forall~ \pi_{m}(\cdot):\mathcal{K}_m\rightarrow \mathcal{K}_m,\nonumber\\ \forall~ m\in\mathcal{M},~~
\forall~ \pi_{0}(\cdot):\mathcal{M}\rightarrow \mathcal{M}.
\end{eqnarray}
Furthermore, let the mapping function from $\left(\mathbf{H}, \mathbf{W}^{(r-1)}\right)$ to $\mathbf{W}^{(r)}$
given by \eqref{CAr} as $g^{(r)}(\cdot,\cdot), \forall~ r\in\mathcal{R}\triangleq\mathcal\{1,2,\ldots,R\}$.
According to the definition of $V\left(\mathbf{H},\mathbf{W}\right)$ in \eqref{SINR3},
we have
\begin{eqnarray}\label{D3}
V\left(\mathbf{H},\mathbf{W}^{(r-1)}\right)&=&V\left(\tilde{\mathbf{H}},\tilde{\mathbf{W}}^{(r-1)}\right),\nonumber\\
\nabla_{\mathbf{w}_m}V\left(\mathbf{H},\mathbf{W}^{(r-1)}\right)&=&\nabla_{\mathbf{w}_{\pi_0(m)}}V\left(\tilde{\mathbf{H}},\tilde{\mathbf{W}}^{(r-1)}\right),
\nonumber\\\forall~ r\in\mathcal{R}.
\end{eqnarray}
Substituting \eqref{D3} into \eqref{CAr} yields
\begin{eqnarray}\label{D4}
\tilde{\mathbf{W}}^{(r)}=g^{(r)}\left(\tilde{\mathbf{H}}, \tilde{\mathbf{W}}^{(r-1)}\right),~~\forall~ r\in\mathcal{R},
\nonumber\\
\forall~ \pi_{m}(\cdot):\mathcal{K}_m\rightarrow \mathcal{K}_m,~~\forall~ m\in\mathcal{M},~~
\forall~ \pi_{0}(\cdot):\mathcal{M}\rightarrow \mathcal{M}.
\end{eqnarray}
\textcolor{black}{Define $\dot{g}^{(R)}(\cdot,\cdot,\cdot)$ as
\begin{eqnarray}\label{D5}
\mathbf{W}^{(R)}&=&g^{(R)}\left(
\mathbf{H}, g^{(R-1)}\left(
\cdots g^{(1)}\left(
\mathbf{H}, g^{\text{sc}}\left(\mathbf{H},\boldsymbol{\alpha}, \boldsymbol{\lambda}
\right)
\right)
\right)
\right)\nonumber\\
&\triangleq&\dot{g}^{(R)}\left(\mathbf{H},\boldsymbol{\alpha}, \boldsymbol{\lambda}\right).
\end{eqnarray}
Substituting \eqref{D2} and \eqref{D4} into \eqref{D5}, we have
\begin{eqnarray}\label{D6}
\tilde{\mathbf{W}}^{(R)}=\dot{g}^{(R)}\left(\tilde{\mathbf{H}},\tilde{\boldsymbol{\alpha}}, \tilde{\boldsymbol{\lambda}}\right), ~~
\forall~ \pi_{m}(\cdot):\mathcal{K}_m\rightarrow \mathcal{K}_m,\nonumber\\ \forall~ m\in\mathcal{M},~~
\forall~ \pi_{0}(\cdot):\mathcal{M}\rightarrow \mathcal{M},
\end{eqnarray}
which yields the HPE property in \eqref{HPE2}.
}



\bibliographystyle{IEEEtran}
\bibliography{IEEEabrv,LiYang}

\end{document}